\documentclass[a4paper,UKenglish]{lipics}
 
\usepackage{microtype}

\usepackage{url}

\usepackage{amsmath}
\usepackage{ upgreek }          
\usepackage{xspace}
\usepackage{color}
\usepackage{multirow}
\usepackage{proof}
\usepackage{esvect}             
\usepackage{wrapfig}
\usepackage{tikz}
\usetikzlibrary{patterns,positioning,calc}
\usepackage{graphicx}

\usepackage{array}

\newcommand{\Spec}{\textsf{Spec}\xspace}
\newcommand{\Apte}{\textsf{Apte}\xspace}

\newcommand{\ignore}[1]{}

\newcommand{\statequiv}{\sim}

\makeatletter
\def\rightarrowfillstar@{\arrowfill@\relbar\relbar{\rightarrow\smash{^*}}}
\newcommand{\xrightarrowstar}[2][]{\ext@arrow
  0{13}{15}8\rightarrowfillstar@{#1}{#2}}
\newcommand{\lrstep}{\@ifstar{\xrightarrowstar}{\xrightarrow}}
\makeatother

\makeatletter
\def\Rightarrowfillstar@{\arrowfill@=={\Rightarrow\smash{^*}}}
\newcommand{\NewxRightarrowstar}[2][]{\ext@arrow
  0{13}{15}8\Rightarrowfillstar@{#1}{#2}}
\newcommand{\NewxRightarrow}[2][]{\ext@arrow
  0{13}{15}8\Rightarrowfill@{#1}{#2}}
\newcommand{\LRstep}{\@ifstar{\NewxRightarrowstar}{\NewxRightarrow}}
\makeatother

\makeatletter
\newcommand{\xRightarrow}[2][]{\ext@arrow 0359\Rightarrowfill@{#1}{#2}}
\makeatother

\usepackage{centernot}

\newcommand{\trip}[3]{(#1;#2;#3)}
\newcommand{\refer}{\mapsto} 
\newcommand{\proc}[2]{(#1;#2)}
\newcommand{\procc}[3]{(#1;#2;#3)}
\newcommand{\dom}{\mathrm{dom}}
\newcommand{\tr}{\mathsf{tr}}

\newcommand{\wfoc}{\varnothing}

\newcommand{\X}{\mathcal{X}}
\newcommand{\W}{\mathcal{W}}
\newcommand{\T}{\mathcal{T}}
\newcommand{\N}{\mathcal{N}}
\newcommand{\q}{\mathcal{Q}}
\newcommand{\p}{\mathcal{P}}
\newcommand{\E}{\mathsf{E}}

\newcommand{\ie}{\emph{i.e.,}\xspace}
\newcommand{\eg}{\emph{e.g.,}\xspace}
\newcommand{\etc}{etc.\xspace}

\usepackage{placeins}

\makeatletter
\newcommand{\rightdoublearrow}{%
  \rightarrow\mkern-10mu\protect\joinrel\rightarrow}

\newcommand{\rightdoublearrowfill@}
  {\arrowfill@\relbar\relbar\rightdoublearrow}
\newcommand{\mapstodoublearrowfill@}
  {\arrowfill@{\mapstochar\relbar}\relbar\rightdoublearrow}
\newcommand{\xrightdoublearrow}[2][]
  {\ext@arrow 3{15}59\rightdoublearrowfill@{#1}{#2}}
\newcommand{\xmapstodoublearrow}[2][]
  {\ext@arrow 3{15}59\mapstodoublearrowfill@{#1}{#2}}
\makeatother

\newcommand{\Dep}{\Dep}

\newcommand{\eqtb}[1]{\equiv_{#1}}

\makeatletter
\newcommand{\xMapsto}[2][]{\ext@arrow 0599{\Mapstofill@}{#1}{#2}}
\def\Mapstofill@{\arrowfill@{\Mapstochar\Relbar}\Relbar\Rightarrow}
\makeatother


\newcommand{\lex}{\mathsf{lex}}
\newcommand{\congr}{\rightsquigarrow}

\newcommand{\test}[3]{\mathtt{if}\ #1\ \mathtt{then}\ #2\ \mathtt{else}\ #3}
\newcommand{\testt}[2]{\mathtt{if}\ #1\ \mathtt{then}\ #2}
\newcommand{\Out}{\mathtt{out}}
\newcommand{\In} {\mathtt{in}}
\newcommand{\Ses} {\mathtt{sess}}
\newcommand{\Foc} {\mathtt{foc}}
\newcommand{\Rel} {\mathtt{rel}}
\newcommand{\Par} {~|~}
\newcommand{\OutS}[1]{\mathtt{out}_{#1}}
\newcommand{\InS}[1]{\mathtt{in}_{#1}}
\newcommand{\BangS}[1]{!^{#1}}

\newcommand{\Para}{\mathtt{par}}
\newcommand{\Zero}{\mathtt{zero}}
\newcommand{\defoc}[1]{\lfloor {#1} \rfloor}
\newcommand{\foc}[1]{\lceil {#1} \rceil}
\newcommand{\deco}[1]{\defoc{#1}}   

\newcommand{\ok}{\mathsf{ok}}

\newcommand{\enc}[2]{\mathsf{enc}(#1,#2)}
\newcommand{\dec}[2]{\mathsf{dec}(#1,#2)}

\newcommand{\bc}{\mathsf{bc}}    
\newcommand{\fc}{\mathsf{fc}}    

\newcommand{\inpar}{\mathrel{\text{\textbardbl}}} 
\newcommand{\inseq}{\rightleftharpoons} 

\newcommand{\ordlex}{\prec_\lex}

\usepackage{graphicx}
\let\oldlrstep\lrstep
\renewcommand{\lrstep}[1]{%
  \mathrel{\raisebox{-0.5pt}[5pt]{%
      $\oldlrstep{\raisebox{-1.2pt}[2pt][1pt]{\scalebox{0.7}{$#1$}}}$%
  }}%
}
\newcommand{\lrstepannot}[2]{%
  \mathrel{\raisebox{-0.5pt}[5pt]{%
      $\oldlrstep{\raisebox{-1.2pt}[2pt][1pt]{\scalebox{0.7}{$#1$}}}$%
  \raisebox{-1.5pt}{\scalebox{0.7}{$#2$}}%
  }}%
}
\renewcommand{\lrstepannot}[2]{\lrstep{#1}_{#2}}
\newcommand{\sint}[1]{\lrstep{#1}}                  
\newcommand{\sintc}[1]{\lrstepannot{#1}{c}}         
\newcommand{\sintd}[1]{\lrstepannot{#1}{r}}         
\newcommand{\sinta}[1]{\lrstepannot{#1}{a}}         

\newcommand{\fsint}[1]{\oldlrstep{#1}}                  
\newcommand{\fsintc}[1]{\oldlrstep{#1}_{c}}         
\newcommand{\fsinta}[1]{\oldlrstep{#1}_{a}}         

\newcommand{\eint}{\approx}         
\newcommand{\eintc}{\approx_c}      
\newcommand{\eintd}{\approx_r}      
\newcommand{\einta}{\approx_a}      

\newcommand{\estat}{\sim}           

\usepackage{esvect}
\newcommand{\vect}[1]{\vv{#1}}

\newcommand{\constrd}[2]{{#2 \triangleright #1}}

\renewenvironment{proof}[1][Proof]{{\noindent\emph{#1.}\;}}{{\qed}\medskip}

\newcommand{\sk}{\mathsf{sk}} 
\newcommand{\skl}{\mathsf{skl}} 

\newcommand{\enab}{\mathsf{enable}}

\newcommand{\ordS}{<}           
\newcommand{\ordSe}{\leq}       

\newcommand{\loc}[2]{[#1]^{#2}} 

\theoremstyle{plain}\newtheorem{proposition}[theorem]{Proposition}
\theoremstyle{plain}\newtheorem{property}[theorem]{Property}

\theoremstyle{plain}\newtheorem*{proposition*}{Proposition}
\theoremstyle{plain}\newtheorem*{theorem*}{Theorem}
\theoremstyle{plain}\newtheorem*{lemma*}{Lemma}

\newcommand{\restatablelemma}[2]{
  \expandafter\newcommand\csname statelemma#1\endcsname{#2}
  \begin{lemma}\label{#1}
    \expandafter\csname statelemma#1\endcsname{}
  \end{lemma}
}
\newcommand{\restatelemma}[1]{
  \begin{lemma*}[\ref{#1}]
    \expandafter\csname statelemma#1\endcsname{}
  \end{lemma*}
}

\newcommand{\restatableproposition}[2]{
  \expandafter\newcommand\csname stateproposition#1\endcsname{#2}
  \begin{proposition}\label{#1}
    \expandafter\csname stateproposition#1\endcsname{}
  \end{proposition}
}
\newcommand{\restateproposition}[1]{
  \begin{proposition*}[\ref{#1}]
    \expandafter\csname stateproposition#1\endcsname{}
  \end{proposition*}
}

\newcommand{\restatabletheorem}[2]{
  \expandafter\newcommand\csname statetheorem#1\endcsname{#2}
  \begin{theorem}\label{#1}
    \expandafter\csname statetheorem#1\endcsname{}
  \end{theorem}
}
\newcommand{\restatetheorem}[1]{
  \begin{theorem*}[\ref{#1}]
    \expandafter\csname statetheorem#1\endcsname{}
  \end{theorem*}
}

\newif\ifnocomment
\newif\ifnotodo
\newif\ifnoproof
\newif\ifnohide
\newif\ifnodel
\nocommenttrue
\notodotrue
\noprooftrue                   
\nohidetrue
\nodeltrue

\newcommand{\squelette}{\sigma}

\nohidefalse
\nodelfalse
\nocommentfalse 
\notodofalse

\title{Partial Order Reduction for Security Protocols}

\author{David Baelde}
\author{St\'ephanie Delaune}
\author{Lucca Hirschi}
\affil{LSV, ENS Cachan \& CNRS\\
\texttt{\{baelde,delaune,hirschi\}@lsv.ens-cachan.fr}
}
\authorrunning{D. Baelde, S. Delaune, and L. Hirschi} 

\Copyright{David\ Baelde, Stéphanie\ Delaune, and Lucca\ Hirschi}

\serieslogo{}
\EventShortName{}

\begin{document}

\maketitle



\begin{abstract}
  Security protocols are concurrent processes that communicate
  using cryptography with the aim of achieving various security
  properties.
  Recent work on their formal verification
  has brought procedures and tools for deciding trace
  equivalence properties (\emph{e.g.} anonymity, unlinkability, vote
  secrecy) for a bounded number of sessions. 
  However, these procedures are based on a naive symbolic exploration of all traces
  of the considered processes which, unsurprisingly, greatly limits
  the scalability and practical impact of the verification tools.
 
  In this paper, we overcome this difficulty by developing partial order
  reduction techniques for the verification of security protocols.
  We provide reduced transition systems that
  optimally eliminate redundant traces, and which are adequate for
  model-checking trace equivalence properties of protocols by means
  of symbolic execution.
  We have implemented our reductions in the tool \Apte, and demonstrated
  that it achieves the expected speedup on various protocols.
\end{abstract}

\section{Introduction}

\newcommand{\myparagraph}[1]{\noindent\emph{#1}~}

Security protocols are concurrent processes that use
various cryptographic primitives in order to achieve security properties
such as secrecy, authentication, anonymity, unlinkability, \emph{etc.}
They involve a high level of concurrency and are difficult to analyse by hand. 
Actually, 
many protocols have been shown to be flawed several years after their
publication (and deployment). 
This has lead to a flurry of research on formal
verification of protocols.


A successful way of representing protocols is to use variants
of the $\pi$-calculus, whose labelled transition systems naturally express
how a protocol may interact with a (potentially malicious) environment
whose knowledge increases as more messages are exchanged over the
network.
Some security properties (\eg\ secrecy, authentication) are then described
as reachability properties, while others (\eg\ unlinkability, anonymity)
are expressed as trace equivalence properties. In order to decide such
properties, a reasonable assumption is to bound the number of protocol
sessions, thereby limiting the length of execution traces.
Even under this assumption, infinitely many traces remain, since
each input may be fed infinitely many different messages.
However, symbolic execution and dedicated constraint solving procedures
have been devised to provide decision procedures for
reachability~\cite{MS02,CCLZ-TOCL} and,
more recently, 
equivalence
properties~\cite{Tiu-csf10,cheval-ccs2011}.
Unfortunately, the resulting 
tools, especially those  for checking equivalence
(\eg \Apte~\cite{Cheval-tacas14}, \Spec~\cite{SPEC}),
have a very limited practical impact because they scale very
badly. 
This is not surprising since they treat concurrency in a very naive way,
exploring all possible symbolic interleavings of concurrent actions.



\myparagraph{Contributions.}
We develop partial order reduction (POR)
techniques~\cite{Peled98,ModelCheckingBook,godefroid1996partial}
for trace equivalence checking of security protocols.
Our main challenge is to
do it in a way that is compatible with symbolic execution:
we should provide a reduction that is effective when messages remain
unknown, but leverages information about messages when it is inferred
by the constraint solver.
We achieve this by refining interleaving semantics 
in two steps, gradually eliminating redundant traces.
The first refinement, called \emph{compression}, uses the notion of
polarity~\cite{Andreoli92} to impose a simple strategy on traces.
It does not rely on data analysis at all and can easily be used as a
replacement for the usual semantics in verification algorithms.
The second one, called \emph{reduction}, takes data into account and
achieves optimality in eliminating redundant traces.
In practice, the reduction step can be implemented in an approximated
fashion, through an extension of constraint resolution procedures.
We have done so in the tool \Apte,
showing that our theoretical results do translate to significant
practical optimisations.

\myparagraph{Outline.}
We consider in Section~\ref{sec:model} a rich process algebra for
representing security protocols.
It supports arbitrary cryptographic primitives,
and even includes a replication operator
suitable for modelling unbounded numbers of sessions.
Thus, we do not restrict to a particular fragment
for which a decision procedure exists,
but show the full scope of our theoretical results.
We give in Section~\ref{sec:prelim} an \emph{annotated} semantics that
will facilitate the next technical developments. We then define our
\emph{compressed} semantics in Section~\ref{sec:compression} and the
\emph{reduced} semantics in Section~\ref{sec:reduction}. In both sections,
we first restrict the transition system, then show that
the restriction is adequate for checking trace equivalence under
some action-determinism condition.
We finally discuss how these results can be lifted to the symbolic
setting in Section~\ref{sec:application}. Specifically, we describe
how we have implemented our techniques in \Apte, and we present experimental
results showing that the optimisations are fully effective in practice.
We discuss related work in Section~\ref{sec:relWork}, and conclude
in Section~\ref{sec:conclusion}.
Complete proofs are given in appendices.

\section{Model for security protocols}
\label{sec:def}
\label{sec:model}

\newcommand{\Ch}{\mathcal{C}}
\newcommand{\Bang}[1]{\textbf{\large !} \nu #1}
\newcommand{\bang}[3]{\mathbf{\large !}^{#1}_{#2} #3}
\newcommand{\obs}{\mathsf{obs}}       
\newcommand{\subst}[2]{\{{#1}/{#2}\}}


In this section we introduce our process algebra, which is a variant
of the applied $\pi$-calculus~\cite{AbadiFournet2001} that has been 
designed with the aim of modelling cryptographic protocols. 
Processes can exchange complex messages, represented by terms
quotiented by some equational theory.

One of
the key difficulties in the applied $\pi$-calculus is to model
the knowledge of the environment, seen as an 
attacker who listens to network communication and may also inject
messages.
One has to make a distinction between the content of a message (sent by the
environment) and the way the message has been created (from knowledge 
available to the environment). 
While the distinction between messages and recipes
came from security applications, it is naturally of much broader interest,
as it gives a precise, intentional content to labelled transitions that we
exploit to analyse data dependencies.

We study a process algebra that may seem quite restrictive:
we forbid internal communication and private channels.
This is however reasonable when studying security protocols faced with the usual
omnipotent attacker.
 In such a setting, we end up considering the worst-case scenario
where any communication 
has to be made via the
 environment.

\subsection{Syntax}

We assume a number of disjoint and infinite sets:
a set $\Ch$ of \emph{channels}, whose elements are denoted by $a$, $b$, $c$;
a set $\N$ of \emph{private names} or \emph{nonces}, denoted by $n$ or $k$;
a set $\X$ of \emph{variables}, denoted by $x$, $y$, $z$ as usual;
and a set $\W$ of \emph{handles}, denoted by $w$ and
used for referring to previously output terms.
Next, we consider a signature $\Sigma$ consisting of a finite set of 
function symbols together with their arity.
Terms over~$S$, written $\T(S)$, are inductively generated from
$S$ and function symbols from $\Sigma$.
When $S \subseteq \N$, elements of $\T(S)$ are called \emph{messages}.
When $S \subseteq \W$, they are called \emph{recipes} and 
written $M$, $N$. Intuitively, recipes express how a message has been
derived by the environment from the messages obtained so far.
Finally, 
we consider an equational theory $\E$ over terms
to assign a 
meaning to function symbols in $\Sigma$. 

\begin{example}
\label{ex:signature}
Let $\Sigma = \{\mathsf{enc}/2,\mathsf{dec}/2, \mathsf{h}/1\}$ and $\E$
be the equational theory induced by the equation  $\dec{\enc{x}{y}}{y} =
x$. Intuitively, the symbols $\mathsf{enc}$ and $\mathsf{dec}$
represent symmetric encryption and decryption, whereas $\mathsf{h}$ is
used to model a hash function.
Now, assume that the environment knows the key~$k$ as well as 
the ciphertext $\enc{n}{k}$,
and that
these two messages
are referred to by handles $w$ and~$w'$.
The environment may decrypt the ciphertext 
with the key~$k$, 
apply the hash function, and  encrypt the
result using~$k$ to get the message $m_0 = \enc{\mathsf{h}(n)}{k}$.  
This computation is modelled using the
\emph{recipe} $M_0 = \enc{\mathsf{h}(\dec{w'}{w})}{w}$.
\end{example}


\begin{definition}
  Processes are defined by the following syntax
where $c,a \in \Ch$, $x \in \X$, $u,v \in \T(\N \cup \X)$, and
$\vect c$ (resp. $\vect{n}$) is a sequence of channels from $\Ch$
(resp. names from~$\N$).
\[
 P,Q  \;::=\;
 0 \;\mid\; (P\Par Q) \;\mid\;
 \In(c,x).P \;\mid\; \Out(c,u).P \;\mid\;
 \test{u=v}{P}{Q} \;\mid\;
 \bang{a}{\vect c, \vect n}{P}
\]
\end{definition}

The last construct combines replication with channel and name restriction:
$\bang{a}{\vect c, \vect n}{P}$ may be read as
$!(\nu \vect{c}.\Out(a,\vect{c}). \nu \vect{n}. P)$
in standard applied $\pi$-calculus. Our goal with this compound construct is 
to support replication in a way that is not fundamentally incompatible
with the action-determinism condition which we eventually impose on our processes.
This is achieved here by advertising on the public channel $a$ any new copy of
the replicated process. At the same time, we make public the new channels
$\vect c$ on which the copy may operate --- but not the new names $\vect n$.
While it may seem restrictive, this style is actually natural
for security protocols where the attacker knows exactly to whom
he is sending a message and from whom he is receiving, \eg via IP addresses.
%

We shall only consider \emph{ground} processes, where
each variable is bound by an input.
We denote by $\fc(P)$ and $\bc(P)$ the set of \emph{free} and \emph{bound
channels} of~$P$.

\newcommand{\h}{\mathsf{h}}
\begin{example}
\label{ex:process}
The process $P_0$ models an agent who sends the ciphertext
$\mathsf{enc}(n,k)$, and then waits for an input on~$c$. In case the
input has the expected form, 
the constant $\mathsf{ok}$ is emitted.\\[2mm]
\null\hfill $P_0=\Out(c,\enc{n}{k}).\In(c,x).\test{\dec{x}{k} =
  \h(n)}{\Out(c,\mathsf{ok}).0}{0}$ \hfill\null

\smallskip{}

\noindent
The processes $P_0$ as well as $\bang{a}{c,n}P_0$  are  ground.
We have that  $\fc(P_0)=\{c\}$ and $\bc(P_0)=\emptyset$ whereas
$\fc(\bang{a}{c,n}P_0) = \{a\}$ and $\bc(\bang{a}{c,n}P_0) = \{c\}$.
\end{example} 

\subsection{Semantics}
\label{subsec:semantics}
We only consider processes that are \emph{normal}
  w.r.t. \emph{internal reduction} $\congr$ defined as follows:
\[
\begin{array}{c}
\left.
\begin{array}{ccc}
\test{u=v}{P}{Q} \congr P \mbox{ \it when } u =_\E
v  &\;\;\;\;&P \Par Q \, \congr\, P' \Par Q\\
\test{u=v}{P}{Q}\congr Q \mbox{ \it when } u\neq_\E v  & &Q \Par P \,\congr\, Q \Par P'
\end{array}
\right\}
\begin{array}{l} \hspace{-0.1cm} {\footnotesize{ \it when\, } P \hspace{-0.1cm}\congr
    \hspace{-0.1cm}  P' } \\
\end{array}
\\[2mm]
\begin{array}{ccccc}
(P_1 \Par P_2) \Par P_3 \,\congr\, P_1\Par
  (P_2 \Par P_3)& \;\;\;\;&   P\Par 0 \congr P & \;\;\;\;&0\Par P \,\congr\, P
\end{array}
\end{array}
\]
  Any process in normal form built from parallel composition can be
  uniquely written as $P_1 \Par (P_2 \Par (\ldots \Par P_n))$
  with $n\geq 2$, which we note $\Pi_{i=1}^n P_i$,
  where each process~$P_i$ is neither a parallel composition nor the process
  $0$.



  We now define our labelled transition system.
  It deals with \emph{configurations} (denoted by~$A$, $B$) which
  are pairs $(\p;\Phi)$ where
    $\p$ is a multiset of ground processes and
    $\Phi$, called the \emph{frame}, is a substitution mapping
      handles to messages that
      have been made available to the environment.
 Given a configuration
  $A$, $\Phi(A)$ denotes its second component.
  Given a frame~$\Phi$, $\dom(\Phi)$ denotes its domain.
\[
  \begin{array}{lcl}
    \mbox{\sc In} &&
    \proc{\{\In(c,x).Q\}\uplus\p}{\Phi} 
    \lrstep{\In(c,M)} 
    \proc{\{Q\subst{M\Phi}{x}\}\uplus\p}{\Phi}  \;\;\;\;\;\;\hfill{\mbox{$M \in \T(\dom(\Phi))$}}
    \\[2mm]
    \mbox{\sc Out} &&
    \proc{\{\Out(c,u).Q\}\uplus\p}{\Phi}
    \lrstep{\Out (c,w)}  \proc{\{Q\}\uplus\p}{\Phi\cup\{w\refer
      u\}} 
    \hfill   \;\;\;\;    {\mbox{$w\in\W$ fresh}}\\[2mm]
    \mbox{\sc Repl} &&
    \proc{\{\bang{a}{\vect{c}, \vect{n}}{P}\} \uplus \p}{\Phi} 
    \lrstep{{\Ses(a,\vect{c})}} 
    \proc{\{P; \bang{a}{\vect c, \vect n}{P}\}\uplus\p}{\Phi} \hfill
    {\mbox{$\vect{c}, \vect{n}$  fresh}}\\[2mm]
    \mbox{\sc Par} &&
    \proc{\{\Pi_{i=1}^n P_i\}\uplus\p}{\Phi}
    \lrstep{\tau} 
    \proc{\{P_1,\ldots,P_n\}\uplus\p}{\Phi}\\[2mm]
    \mbox{\sc Zero} & &
    \proc{\{0\}\uplus\p}{\Phi}
    \lrstep{\tau} \proc{\p}{\Phi}
  \end{array}
  \]

  Rule \textsc{In} expresses that an input process may receive any message that
  the environment can derive from the current frame. In rule \textsc{Out}, the
  frame is enriched with a new message. 
The last two rules simply translate the parallel structure of processes
  into the multiset structure of the configuration.
As explained above, rule \textsc{Repl}
  combines the replication of a process together with the creation of
  new channels and nonces.
The channels~$\vect c$ are implicitly made public, but the newly created
names $\vect n$ remain private.
Remark that channels $\vect c$ and names $\vect n$ must be
fresh, \ie they do not appear free in the original configuration.
As usual, freshness conditions do not block executions:
it is always
possible to rename bound channels~$\vect c$ and names $\vect n$
of a process $\bang{a}{\vect c,\vect n}{P}$ before applying {\sc Repl}.
We denote by $\bc(\tr)$ the bound channels of a trace $\tr$, \ie all
the channels that occur in second argument of an action $\Ses(a,\vect
c)$ in $\tr$,
and we consider traces
where channels are bound at most once.

\begin{example}
\label{ex:semantics}
Going back to Example~\ref{ex:process} with
$\Phi_0 = \{w_1 \mapsto k\}$,
we have that:\\[2mm]
\null\hfill
$\proc{\{\bang{a}{c,n}P_0\}}{\Phi_0} \sint{\Ses(a,c)} 
\lrstep{\Out(c,w_2)} 
\lrstep{\In(c,M_0)}
\proc{\{\Out(c,\ok).0;\bang{a}{c,n}P_0\}}{\Phi}$\hfill\null
\\[2mm]
where 
$\Phi = \{w_1 \mapsto k, w_2 \mapsto \enc{n}{k}\}$ and  $M_0 =
\enc{\h(\dec{w_2}{w_1})}{w_1}$.
\end{example}

 \subsection{Equivalences}

We are concerned with trace equivalence,
which is 
used~\cite{kostas-csf10,DKR-lncs6000}
to model anonymity, untraceability, strong secrecy, \etc
Finer behavioural equivalences, \eg weak bisimulation, appear to be too strong
with respect to what an attacker can really observe.
%
Intuitively, two configurations are trace equivalent if the attacker 
cannot tell whether he is interacting with one or the other.
To make this formal, we  introduce a notion of equivalence 
between  frames. 

\begin{definition}
Two frames $\Phi$ and $\Phi'$ are in \emph{static equivalence},
written $\Phi
\statequiv \Phi'$, when $\dom(\Phi) = \dom(\Phi')$,
and:
$M\Phi =_\E N\Phi \;\Leftrightarrow \; M\Phi' =_\E N\Phi'
\mbox{ for any terms $M, N \in \T(\dom(\Phi))$}.
$
\end{definition}

\begin{example}
\label{ex:static}
Continuing Example~\ref{ex:semantics}, consider $\Phi' = \{w_1
\mapsto k', w_2 \mapsto \enc{n}{k}\}$. 
The test ${\enc{\dec{w_2}{w_1}}{w_1} = w_2}$
is true in $\Phi$ but not in~$\Phi'$, thus $\Phi
\not\statequiv \Phi'$. 
\end{example}

\noindent We then define $\obs(\tr)$ to be the subsequence of $\tr$
obtained 
by erasing $\tau$ actions.

\begin{definition}
  Let $A$ and $B$ be two configurations. We say that $A \sqsubseteq B$ when,
  for any $A \lrstep{\tr} A'$ such that $\bc(\tr)\cap\fc(B)=\emptyset$,
  there exists $B \lrstep{\tr'} B'$ such
  that $\obs(\tr) = \obs(\tr')$ and ${\Phi(A') \statequiv \Phi(B')}$.
  They are \emph{trace equivalent}, written $A \approx B$, 
  when $A \sqsubseteq B$ and $B \sqsubseteq A$.
\end{definition}

In order to lift our optimised semantics to trace equivalence,
we will require configurations to be {\em action-deterministic}.
This common assumption in POR techniques~\cite{ModelCheckingBook}
is also reasonable in the context of security protocols, where the 
attacker knows with whom he is communicating.


\begin{definition}
A configuration $A$ is {\em action-deterministic} if whenever $A \lrstep{\tr}
\proc{\p}{\Phi}$, and 
$P,Q$ are two elements of~$\p$, we have that 
$P$ and $Q$ cannot perform an observable action of the same nature
($\In$, $\Out$, or $\Ses$) on the same channel (\ie if both actions
are of same nature, their first argument has to differ).
\label{def:action-det}
\end{definition}

\section{Annotated semantics}
\label{sec:prelim}


We shall now define an intermediate semantics whose transitions are equipped 
with more informative actions. The annotated actions will notably feature 
\emph{labels} $\ell \in \mathbb{N}^{*}$ indicating from which concurrent processes
they originate.
A {\em labelled action} will be written $\loc{\alpha}{\ell}$ where $\alpha$
is an action and $\ell$ is a label.
Similarly, a {\em labelled process} will be written $\loc{P}{\ell}$.
When reasoning about trace equivalence between two configurations, it will be 
crucial to maintain a consistent labelling between configurations along the execution.
In order to do so, we 
define \emph{skeletons of observable actions},
which are of the form
$\InS c$, $\OutS c$ or $\BangS a$ where $a,c\in\Ch$,
and we 
assume a total ordering over  those skeletons,
denoted $\ordS$ with $\ordSe$ its reflexive closure.
Any process that is neither $0$ nor a parallel composition
induces a skeleton corresponding to its toplevel connective,
and we denote it by $\sk(P)$.

We define in Figure~\ref{fig:sem-ann} the {\em annotated semantics}
$\sinta{}$ over configurations whose processes are labelled.
In \textsc{Par}, note that $\sk(P_i)$ is well defined as $P_i$ cannot be a zero
or a parallel composition.
Also note
that the label of an action is always that of the active process
in that transition.
More importantly,
the annotated transition system does not restrict the executions
of a process but simply annotates them with labels, and replaces
$\tau$ actions by more descriptive actions.

\begin{figure}[tpb]
{  \centering
$
 \begin{array}{l}
    \mbox{\sc In\slash Out} \;\;
 \infer[\alpha\in\{\In(\_,\_);\Out(\_,\_)\}]
    {\proc{\{\loc{P}{\ell}\}\uplus\p}{\Phi}
      \fsinta{\;\loc{\alpha}{\ell}\;}
      \proc{\{\loc{P'}{\ell}\}\uplus\p}{\Phi'}}
    {\proc{\{P\}\uplus\p}{\Phi}\fsint{\;\alpha\;}
      \proc{\{P'\}\uplus\p}{\Phi'}}
    \\ [3mm]
    \mbox{\sc Repl} \;\;
    \proc{\{\loc{\bang{a}{\vect{c}, \vect{n}}{P_0}}{\ell}\} \uplus \p}{\Phi} 
    \fsinta{\loc{\Ses(a,\vect{c})}{\ell}}
    \proc{\{\loc{P_0}{\ell \cdot 1},
      \loc{\bang{a}{\vect c, \vect n}{P_0}}{\ell \cdot 2}
      \}\uplus\p}{\Phi} \;\;
 \mbox{$\vect c,\vect n$  fresh}\\[3mm]
    \mbox{\sc Par} \;\;
    \proc{\{\loc{\Pi_{i=1}^n P_i}{\ell}\}\uplus\p}{\Phi}
    \fsinta{\loc{\Para(\squelette_{\pi(1)};\ldots;\squelette_{\pi(n)})}{\ell}}
    \proc{\{\loc{P_{\pi(1)}}{\ell \cdot
        1},\ldots,\loc{P_{\pi(n)}}{\ell \cdot n}\}\uplus\p}{\Phi}\\[1mm]
    \multicolumn{1}{r}{$\mbox{$\squelette_i = \sk(P_i)$ and
$\pi$ is a permutation
      over $[1,...n]$ such that
      $\squelette_{\pi(1)} \leq \ldots \leq \squelette_{\pi(n)}$}$}\\[3mm]
    \mbox{\sc Zero}  \;\;
 \hspace{1cm}
    \proc{\{\loc{0}{\ell}\}\uplus\p}{\Phi}
    \fsinta{\loc{\Zero}{\ell}} \proc{\p}{\Phi}
  \end{array}
$

}
\caption{Annotated semantics}
\label{fig:sem-ann}
\end{figure}

We now define how to extract sequential dependencies from labels,
which will allow us to analyse concurrency in a trace without
referring to configurations.
\begin{definition}
Two labels are \emph{dependent} if one is a prefix of the other.
We say that the labelled actions $\alpha$ and $\beta$
are \emph{sequentially dependent} when their labels are dependent,
and \emph{recipe dependent} when
$\{\alpha, \beta\} = \{\loc{\In(c,M)}{\ell},\loc{\Out(c',w)}{\ell'}\}$
with $w$ occurring in $M$.
They are \emph{dependent} when they are either sequentially or recipe
dependent. Otherwise, they are \emph{independent}.
\end{definition}
\begin{definition}
  A configuration $\proc{\p}{\Phi}$ is \emph{well labelled} if
  $\p$ is a multiset of labelled processes such that
  two elements of $\p$ have independent labels.
\label{def:well-label}
\end{definition}

Obviously, any unlabelled configuration may be well labelled.
Further, it is easy to see that well labelling is preserved by $\sinta{}$.
Thus, we shall implicitly assume to be working with well labelled
configurations in the rest of the paper. Under this assumption,
we obtain the following fundamental lemma.
\restatablelemma{lem:perm}{
  Let $A$ be a (well labelled) configuration,
  $\alpha$ and $\beta$ be two independent labelled actions.
  We have
  $A \sinta{\alpha . \beta} A'$
  if, and only if,
  $A \sinta{\beta . \alpha} A'$.
}

\paragraph{Symmetries of trace equivalence.}
We will see that, when checking $A \eint B$ for action-deterministic
configurations, it is sound to require that $B$ can perform all traces of $A$
in the annotated semantics (and the converse). In other words, labels and detailed non-observable 
actions $\Zero$ and $\Para(\sigma_1 \ldots \sigma_n)$ are actually relevant for trace equivalence.
Obviously, this can only hold if $A$ and $B$ are labelled consistently.
In order to express this, we extend $\sk(P)$ to parallel and zero processes:
we let their skeletons be the associated action in the annotated semantics.
Next, we define the labelled skeletons by
$\skl(\loc{P}{\ell}) = \loc{\sk(P)}{\ell}$.
When checking for equivalence of $A$ and $B$, we shall assume that
$\skl(A)=\skl(B)$, \ie the configurations have the same set of labelled
skeletons. This technical condition is obviously not restrictive in practice.

\begin{example}
\label{ex:out-label}
Let $A = 
\proc{\{\loc{\In(a,x).((\Out(b,m).P_1)\Par P_2)}{0}\}}{\Phi}$ 
with $P_1= \In(c,y).0$ and $P_2=\In(d,z).0$,
and $B$ the configuration obtained from $A$ by swapping $P_1$ and $P_2$.
We have $\skl(A)=\skl(B)=\{\loc{\InS a}{0}\}$.
Consider the following trace:\\[2mm]
\null\hfill
$\tr=\loc{\In(a,\mathtt{ok})}{0}.\loc{\Para(\{\OutS b;\InS d\})}{0}.\loc{\Out(b,w)}{0\cdot1}.
\loc{\In(c,w)}{0\cdot1}.\loc{\In(d,w)}{0\cdot2}$\hfill\null
\\[2mm]
Assuming $\OutS b \ordS\InS d$ and $\mathsf{ok}\in\Sigma$, we
have $A\sinta{\tr}A'$. However, there is no~$B'$
such that $B\sinta{\tr}B'$, for two reasons. First,
$B$ cannot perform the second action
since skeletons of sub-processes of its parallel composition are $\{\OutS b;\InS c\}$.
Second, even if we ignored that mismatch on a non-observable action,
$B$ would not be able to perform the action $\In(c,w)$ with the right label.
Such mismatches can actually be systematically used to show $A \not\eint B$,
as shown next.
\end{example}

\restatablelemma{lem:strong-symmetry}{
Let $A$ and $B$ be two action-deterministic configurations such that
${A \eint B}$ and $\skl(A) = \skl(B)$.
For any execution\\[1mm]
\null\hfill
$A \fsinta{\loc{\alpha_1}{\ell_1}} A_1
\fsinta{\loc{\alpha_2}{\ell_2}}  \ldots
\fsinta{\loc{\alpha_n}{\ell_n}} A_n$
\hfill\null\\[1mm] 
\noindent
with $\bc(\alpha_1.\ldots\alpha_n)\cap\fc(B)=\emptyset$,
there exists an execution\\[1mm]
\null\hfill
$B \fsinta{\loc{\alpha_1}{\ell_1}} B_1
\fsinta{\loc{\alpha_2}{\ell_2}}  \ldots
\fsinta{\loc{\alpha_n}{\ell_n}} B_n$
\hfill\null\smallskip \\[1mm]
\noindent
such that $\Phi(A_i) \estat
\Phi(B_i)$ and $\skl(A_i) =
\skl(B_i)$ for any $1 \leq i \leq n$.
}



\section{Compression}
\label{sec:compression}

Our first refinement of the semantics,
which we call compression, is closely related to
focusing from proof theory~\cite{Andreoli92}: we will assign a polarity to
processes and constrain the shape of executed traces based on those 
polarities. This will provide a first significant reduction
of the number of traces to consider when checking reachability-based 
properties such as secrecy, and more importantly, equivalence-based properties in the
action-deterministic case.

\begin{definition}
A  process $P$ is \emph{positive} if it is of the form
$\In(c,x).Q$, and it is \emph{negative} otherwise.
A multiset of processes~$\p$ is \emph{initial} if it contains only positive
or \emph{replicated} processes,
\ie of the form $\bang{a}{\vect{c}, \vect n}{Q}$.
\end{definition}

\begin{figure}[h]
\[\begin{array}{lcl}
    \mbox{\sc Start\slash In}&  & {\infer[]
    {\trip{\mathcal{P}\uplus{\{P\}}}{\wfoc}{\Phi}\fsintc{\;\Foc(\In(c,M))\;}
      \trip{\mathcal{P}}{P'}{\Phi}}
    {   \p\text{ is initial}
      & \proc{P}{\Phi}\fsinta{\In(c,M)} \proc{P'}{\Phi}
    }}
    \\ [1mm]
    \mbox{\sc Start\slash !} &&  \infer
    {\trip{\mathcal{P}\uplus{\{\bang{a}{\vect c, \vect n}{P}\}}}{\wfoc}{\Phi}\fsintc{\;\Foc(\Ses(a,\vect {c}))\;} 
      \trip{\mathcal{P}\uplus{\{\bang{a}{\vect c, \vect n}{P}\}}}
      {Q}{\Phi}}
    {    \p\text{ is initial}
      & \proc{\bang{a}{\vect c,\vect
          n}{P}}{\Phi}\fsinta{\Ses(a,\vect{c})} \proc{\{ \bang{a}{\vect c,\vect
          n}{P};Q\}}{\Phi} 
    }
    \\ [1mm]
    \mbox{\sc Pos\slash In} &&  \infer[]
    {\trip{\mathcal{P}}{P}{\Phi}\fsintc{\;\In(c,M)\;}
      \trip{\mathcal{P}}{P'}{\Phi}}
      {   
      \proc{P}{\Phi}\fsinta{\In(c,M)} \proc{P'}{\Phi}
    }
    \\ [1mm]
    \mbox{{\sc Neg}} && 
  \infer[\alpha\in\{\Para(\_),\Zero,\Out(\_,\_)\}]
    {\trip{\mathcal{P}\uplus{\{P\}}}{\wfoc}{\Phi}\fsintc{\;\alpha\;}
      \trip{\mathcal{P}\uplus{\mathcal{P}'}}{\wfoc}{\Phi'}}
    { 
      \proc{P}{\Phi}\fsinta{\alpha} \proc{\mathcal{P}'}{\Phi'}
    }
    \\ [1mm]
    \mbox{\sc Release} &&  
    {\trip{\mathcal{P}}{\loc{P}{\ell}}{\Phi}\fsintc{\;\loc{\Rel}{\ell}\;} 
      \trip{\mathcal{P}\uplus{\{\loc{P}{\ell}\}}}{\wfoc}{\Phi}}
    \;\;\mbox{ when  $P$\text{ is negative }
    }
  \end{array}\]
  Labels are implicitly set in the same way as in the annotated semantics.
  {\sc Neg} is made non-branching by imposing an arbitrary
  order on labelled skeletons of available actions.
  \caption{Compressed semantics}
\label{fig:comp-sem:one}
\end{figure}


The compressed semantics (see Figure~\ref{fig:comp-sem:one}) is built upon the annotated semantics.  
It constrains the traces to
follow a particular strategy, alternating between \emph{negative} and
\emph{positive} phases.
It uses enriched configurations of the form
$\trip{\mathcal{P}}{F}{\Phi}$ where $\proc{\mathcal{P}}{\Phi}$ is
a labelled configuration and
$F$ is either a  process (signalling which process is {\em under focus}
 in the positive phase) or~$\wfoc$ (in the negative phase).
The negative phase lasts until the configuration is initial
(\ie unfocused with an initial underlying multiset of processes) 
and in that phase we perform actions that decompose negative non-replicated
processes. This is done using the {\sc Neg} rule, in a completely 
deterministic way.
When the configuration becomes
initial, a positive phase starts:
we choose one process and start executing the actions of that 
process (only inputs, possibly preceded by a new session) without the 
ability to switch to another process of the multiset,
until a negative subprocess is released and we go back to the negative phase.
The active process in the positive phase is said to be \emph{under focus}.
Between any two initial configurations, the compressed semantics
executes a sequence of actions, called \emph{blocks},
of the form
\;$\Foc(\alpha).\tr^+.\Rel.\tr^-$\;
where~$\tr^+$ is a (possibly empty) 
sequence of input actions, whereas $\tr^-$ is a (possibly empty) sequence
of $\Out$, $\Para$, and $\Zero$ actions. 
Note that, except for the choice of recipes,
the compressed semantics is completely non-branching when executing a block.
It may branch only when choosing which block to execute.

\subsection{Reachability}

We now formalise the relationship between traces of the compressed and
annotated semantics. In order to do so,
we translate between configuration and enriched configuration as follows:
\\[2mm]
\null\hfill $\foc{\proc{\mathcal{P}}{\Phi}} =
\trip{\mathcal{P}}{\wfoc}{\Phi}$, \;
$\defoc{\trip{\mathcal{P}}{\wfoc}{\Phi}} = \proc{\mathcal{P}}{\Phi}$ \;
and
$\defoc{\trip{\mathcal{P}}{P}{\Phi}} =
\proc{\mathcal{P}\uplus{\{P}\}}{\Phi}$. \hfill\null

\smallskip{}

Similarly, we map compressed traces to annotated ones:
 \\[2mm]
\null\hfill
 $\defoc{\epsilon}=\epsilon$, \;
 $\defoc{\Foc(\alpha).\tr} = \alpha.\defoc{\tr}$, \;
 $\defoc{\Rel.\tr} = \defoc{\tr}$ \; and
 $\defoc{\alpha.\tr}=\alpha.\defoc \tr$ otherwise.
 \hfill\null

\smallskip{}

We observe  that we can map any execution in the compressed
semantics to an execution in the annotated semantics.
Indeed,
a compressed execution is simply an annotated execution with some
annotations (\ie $\Foc$ and $\Rel$) indicating the start of a
positive/negative phase.

\restatablelemma{lem:reach-sound}{
For any configurations
$A$,  $A'$ and $\tr$,
$A \sintc{\tr} A'$ implies $\defoc{A} \fsinta{\defoc{\tr}} \defoc{A'}$.
}

Going in the opposite direction is more involved.
In general, mapping annotated executions to compressed ones requires
to reorder actions.
Compressed executions also force negative actions to be performed
unconditionally,
which we compensate by considering
{\em complete} executions of a configuration, \ie executions after which no more
action can be performed except possibly the ones that consist in
unfolding a
replication (\ie rule {\sc Repl}). Inspired by the positive
trunk argument of~\cite{saurin}, we show the following lemma. 

\restatablelemma{thm:reach-comp}{
Let $A$, $A'$ be two configurations and $\tr$ be such that $A\sinta{\;\tr\;} A'$ is
complete.
There exists a trace~$\tr_c$, such that
$\defoc{\tr_c}$ can be obtained from $\tr$ by swapping
independent labelled actions, and
$\foc{A}\sintc{\tr_c}\foc{A'}$.
}

\begin{proof}[Proof sketch]
  We proceed by induction on the length of a complete execution
  starting from~$A$.
  If $A$ is not initial, then we need to execute some negative
  action using {\sc Neg}: this action must be present
  somewhere in the complete execution, and we can permute it
  with preceding actions using Lemma~\ref{lem:perm}.
  If $A$ is initial, we analyse the prefix of input and session actions
  and  we  extract a subsequence of that prefix that
  corresponds to a full positive phase.
\end{proof}



\subsection{Equivalence}
\label{sec:comp-equ}
We now define compressed trace equivalence ($\eintc$) and
prove that it coincides with $\approx$. 

\begin{definition}
Let~$A$ and~$B$ be two configurations. We say that $A \sqsubseteq_c B$
when, for any
$A \sintc{\tr} A'$ such that $\bc(\tr)\cap\fc(B)=\varnothing$, 
there exists $B\sintc{\tr} B'$ such that $\Phi(A')\statequiv \Phi(B')$.
They are \emph{compressed trace equivalent}, denoted $A \eintc B$, if 
$A \sqsubseteq_c B$ and $B \sqsubseteq_c A$.
\label{def:eintc}
\end{definition}

Compressed trace equivalence can be more efficiently checked than regular
trace equivalence. Obviously, it explores less interleavings by relying on
$\sintc{}$ rather than $\sint{}$. It also requires that traces of one process
can be played exactly by the other, including details such as non-observable
actions, labels, and focusing annotations.
The subtleties shown in Example~\ref{ex:out-label} are crucial for the 
completeness of compressed equivalence w.r.t. regular equivalence.
Since the compressed semantics forces
to perform available outputs before \emph{e.g.} input actions, some non-equivalences
are only detected thanks to the labels and detailed non-observable actions
of our annotated semantics.

\restatabletheorem{thm:eintc}{
  Let $A$ and $B$ be two  action-deterministic configurations with 
  $\skl(A)=\skl(B)$.
  We have {$A \eint B$} if, and only if, $\foc{A} \eintc \foc{B}$.
}




\begin{proof}[Proof sketch]
  ($\Rightarrow$)
  Consider an execution $\foc{A} \sintc{\tr} A'$.  
  Using Lemma~\ref{lem:reach-sound}, we get  
  ${A} \sinta{\defoc{\tr}} \defoc{A'}$. Then, Lemma~\ref{lem:strong-symmetry} yields
  $B\sinta{\defoc{\tr}}B'$ for some $B'$ such that
  $\Phi(\defoc{A'})\estat\Phi(B')$ and labelled skeletons are equal all along
  the executions. Relying on those skeletons, we show that positive/negative
  phases are  synchronised, and thus 
  $\foc{B}\sintc{\tr}B''$
  for some $B''$
  with $\defoc{B''}=B'$. 
  \\  ($\Leftarrow$)
  Consider an execution $A \sinta{\tr} A'$.  
  We first observe that it suffices to consider only complete executions 
  there. This allows us to get 
a compressed 
  execution $\foc{A} \sintc{\tr_c} \foc{A'}$
  by Lemma~\ref{thm:reach-comp}.
  Since $\foc{A} \eintc \foc{B}$, there
  exists $B'$ such that $\foc{B} \sintc{\tr_c} B'$ with $\Phi(\foc{A'})\estat\Phi(B')$. 
  Thus we have $B \sinta{\deco{\tr_c}} \defoc{B'}$ but also $B \sinta{\tr} \defoc{B'}$
  thanks to Lemma~\ref{lem:perm}.
\end{proof}


\label{rmk:compimpropre}
\myparagraph{Improper blocks.}
Note that blocks of the form $\Foc(\alpha).\tr^+.\Rel.\Zero$ do not bring any new
information to the attacker.
While it would be incorrect to fully ignore such \emph{improper} blocks,
we can show that it is sufficient to consider them
at the end of traces.
We can thus consider a further optimised compressed trace equivalence that
only checks for {\em proper traces},
\ie ones that have at most one improper block and only at the end of trace.
We have also shown that this optimised compressed trace equivalence actually
coincides with ${\eintc}$.




\section{Reduction}
\label{sec:reduction}

Our compressed semantics cuts down interleavings by using a simple
focused strategy. 
However,
this semantics does not analyse data dependency that happen when an input
depends on an output, and is thus unable to exploit the independency of
blocks to reduce interleavings.
We tackle this problem now.

\newcommand{\block}{\mathit{block}}


%
%

\begin{definition}
\label{def:par-seq}
Two blocks $b_1$ and $b_2$ are \emph{independent}, written  $b_1
\inpar b_2$, when all labelled actions $\alpha_1 \in b_1$ and $\alpha_2 \in b_2$
are independent. Otherwise they
are \emph{dependent}, written $b_1 \inseq b_2$.
\end{definition}

Obviously, Lemma~\ref{lem:perm} tells us that independent blocks can
be permuted in a trace without affecting the executability and the
result of executing that trace. But this notion is not very strong
since it considers fixed recipes, which are irrelevant (in the end,
only the derived messages matter) and can easily introduce spurious
dependencies.
Thus we define a stronger notion of equivalence over
traces, which allows permutations of independent blocks but also
changes of recipes that preserve messages. 
During these permutations, we will also require that traces remain {\em plausible},
which is defined as follows:
$\tr$ is plausible if for any input $\In(c,M)$ such that
$\tr=\tr_0.\In(c,M).\tr_2$
then $M\in\T(\W)$ where $\W$ is the set of all handles occurring in $\tr_0$.
Given a block $b$, \emph{i.e.} a sequence of the form
$\Foc(\alpha). \tr^+. \Rel. \tr^-$, we denote by $b^+$ (resp. $b^-$)
the part of $b$ corresponding to the
positive (resp. negative) phase, \ie $b^+ = \alpha.\tr^+$
(resp. $b^-=\tr^-$).
We note $(b_1=_\E b_2)\Phi$ when $b_1^+\Phi =_\E b_2^+\Phi$
and $b_1^-=b_2^-$.

\begin{definition}
\label{def:eqtb}
  Given a frame $\Phi$,
  the relation $\eqtb{\Phi}$ is the smallest equivalence over 
  plausible compressed traces 
such
  that
$\tr.b_1.b_2.\tr' \eqtb{\Phi} \tr.b_2.b_1.\tr'$ when 
  $b_1\mathrel{\inpar} b_2$, and
$\tr.b_1.\tr'\eqtb{\Phi} \tr.b_2.\tr'$ when $(b_1=_\E b_2)\Phi$.
\end{definition}


\restatablelemma{lem:block-swap}{
  Let $A$ and $A'$
  be two initial configurations such that
  $A \sintc{\tr} A'$.
  We have that 
  $A \sintc{\tr'} A'$
  for any $\tr' \eqtb{\Phi(A')} \tr$.
}

We now turn to defining our reduced semantics, which is going to
avoid the redundancies identified above by only executing specific
representatives in equivalence classes modulo $\eqtb{\Phi}$.
More precisely, we shall only execute
minimal traces according to some order, which we now introduce.
We assume an order~${\prec}$ on blocks that is insensitive to recipes,
and such that independent blocks are always strictly ordered
in one way or the other.
We finally define~${\ordlex}$ on compressed traces as the lexicographic
extension of~${\prec}$ on blocks.


In order to incrementally build representatives that are minimal
with respect to ${\ordlex}$, we define a predicate
that expresses whether a block $b$ should be \emph{authorised} after
a given trace~$\tr$.
Intuitively, this is 
the 
case only when, for any block
$b'\succ b$ in $\tr$, dependencies forbid to swap $b$ and $b'$.
We define this with recipe dependencies first, then quantify over all
recipes to capture 
message dependencies.

\begin{definition}
\label{def:dep-analysis}
A block $b$ is authorised after $\tr$, noted $\constrd{b}{\tr}$, when
$\tr = \epsilon$; or $\tr  = \tr_0. b_0$ and either
{\it (i)} $b \inseq b_0$ or 
{\it (ii)} $b \inpar b_0$, $b_0 \prec b$, and
  $\constrd{b}{\tr_0}$.
\end{definition} 


We finally define $\sintd{}$ as the least relation such that:
\[
\begin{array}{llcll}
\mbox{\sc Init} \;& \multirow{2}{*}{\infer[]{A \sintd{\epsilon}
    A}{\phantom{\proc{\p}{\Phi} \sintd{\tr}}}}& &\mbox{\sc Block} \;&
\multirow{2}{*}{\infer[\begin{array}{c}\mbox{if $\constrd{b'}{\tr}$ for all $b'$}\\ \mbox{ with
    $(b'=_\E b)\Phi$}\end{array}]{A\sintd{\tr.b} A'}{
       A\sintd{\tr} \procc{\p}{\wfoc}{\Phi} &
       \procc{\p}{\wfoc}{\Phi} \sintc{b} A'}} \\
&&&
\end{array}
\]
Our reduced semantics only applies to initial configurations:
otherwise, no block can be performed.
This is not restrictive since we can, without loss
of generality, pre-execute
non-observable and output actions that may occur at top level.

\subsection{Reachability}
An easy induction on the compressed trace~$\tr$ allows us to map an execution w.r.t. the reduced semantics to an
execution w.r.t. the compressed semantics. 

\begin{lemma}
  For any configurations $A$ and $A'$,
  $A \sintd{\tr} A'$ implies $A \sintc{\tr} A'$.
\label{lem:reach-red-sound}
\end{lemma}

Next, we show that our reduced semantics
only explores specific representatives.
Given a frame $\Phi$, a plausible trace $\tr$ is
$\Phi$-minimal if it is minimal\footnote{
  Note that minimal traces are not unique, since only labelled skeletons are taken
  into account when comparing actions. However, the redundancy induced by
  the choice of recipes is not a concern in practice as it does not arise
  with the current constraint-based techniques for deciding trace equivalence.
}
in its equivalence class modulo $\eqtb{\Phi}$.

\restatablelemma{lem:min-swap}{
  Let $A$ be an initial configuration
  and $A'=\procc{\p}{\wfoc}{\Phi}$ be a configuration such that $A\sintc{\tr} A'$.
  We have that $\tr$ is $\Phi$-minimal if, and only if,
  $A\sintd{\tr}A'$.
}

\begin{proof}[Proof sketch]
  In order to relate minimality and executability in the reduced semantics,
  let us say that a trace is \emph{bad} if it is of the form
  $\tr.b_0\ldots b_n.b'.\tr'$ where
  $n\geq 0$, there exists a block $b''$ such that
  $(b''=_\E b')\Phi$, we have
  $b_i \inpar b''$ for all $i$, and
  $b_i \prec b'' \prec b_0$ for all $i>0$.
  This pattern is directly inspired by the characterisation of
  lexicographic normal forms by Anisimov and Knuth in trace monoids~\cite{knuth-sorting}.
  We note that a trace that can be executed in the compressed 
  semantics can also be executed in the reduced semantics if, and only if,
  it is not bad. Since the badness of a trace allows to swap $b'$ before
  $b_0$, and thus obtain a smaller trace in the class~$\equiv_\Phi$,
  we show that a bad trace cannot be $\Phi$-minimal (and conversely).
\end{proof}

\subsection{Equivalence}

The reduced semantics 
 induces an equivalence $\eintd$
that we define similarly to the compressed one, and
we then establish its soundness and completeness w.r.t.~$\eintc$.


\begin{definition}
Let~$A$ and~$B$ be two configurations. We say that $A \sqsubseteq_r B$ when,
for every $A \sintd{\tr} A'$ such that $\bc(\tr)\cap\fc(B)=\varnothing$, 
there exists $B\sintd{\tr} B'$ such that $\Phi(A')\statequiv \Phi(B')$.
They are  \emph{reduced trace equivalent}, denoted $A \eintd B$, if 
$A \sqsubseteq_r B$ and $B \sqsubseteq_r A$.
\label{def:eintd}
\end{definition}


\restatabletheorem{thm:redequiv}{
  Let $A$ and $B$ be two initial,
  action-deterministic configurations. 
\\[1mm]
\null\hfill
$A \eintc B \mbox{ if, and only if, }
A \eintd B$\hfill\null
}

\begin{proof}[Proof sketch]
We first prove that $\tr \eqtb{\Phi} \tr'$ iff
  $\tr \eqtb{\Psi} \tr'$ when $\Phi \statequiv \Psi$. 
%
($\Rightarrow$) 
This implication 
is then an easy consequence of Lemma~\ref{lem:min-swap}.
($\Leftarrow$) 
We start by showing that it suffices to consider 
a 
complete execution  $A\sintc{\tr}A'$.
Since $A'$ is initial, by taking $\tr_m$ to be a $\Phi(A')$-minimal
trace associated to $\tr$, we obtain a reduced execution of $A$
leading to $A'$. Using our hypothesis $A\eintd B$,  we obtain that $B\sintd{\tr_m}B'$ with corresponding relations
over frames. We finally conclude that $B\sintc{\tr}B'$ using 
Lemma~\ref{lem:block-swap} and the result stated above.
\end{proof}

\label{rmk:redimpropre}
\myparagraph{Improper blocks.}
  Similarly as we did for the compressed semantics in 
  Section~\ref{sec:compression}, we can further restrict
  ${\eintd}$ to only check proper traces (see Appendix~\ref{sec:impropre}).

\section{Application} \label{sec:application}

We have developed two successive refinements of the concrete semantics of
our process algebra, eventually obtaining a reduced semantics that
achieves an optimal elimination of redundant interleavings.
However, the practical usability of these semantics in algorithms for checking 
the equivalence of replication-free processes is far from immediate: indeed, 
all of our semantics are still infinitely branching, because each input may be 
fed with arbitrary messages.
We now discuss how existing decision procedures based on
symbolic execution~\cite{MS02,CCLZ-TOCL,Tiu-csf10,cheval-ccs2011} can
be modified to decide our optimised equivalences rather than the regular one,
before presenting our implementation and experimental results.

\subsection{Symbolic execution}
\label{subsec:symbolic}
Our compressed semantics can easily be used as a replacement of the regular
one, in any tool whose algorithm is based on a forward exploration of
the set of possible traces.
This modification is very lightweight, and already brings a significant
optimisation.
In order to make use of our final, reduced semantics, we need to enter into
the details of constraint solving. In addition to imposing the compressed
strategy, symbolic execution should be modified to generate \emph{dependency
constraints} 
in order to reflect the data dependencies imposed
by our predicate $\constrd{b}{\tr}$. Actually, the generation of
dependency constraints can be done in a similar way as shown
in~\cite{BDH-post14} (even if the class of processes considered
in~\cite{BDH-post14} is more restrictive). We simply illustrate the
effect of these dependency constraints on a simple example.

\begin{example}
We consider roles
$R_i := \In(c_i,x).\testt{x=\ok}{\Out(c_i,\ok)}$
where $\ok$ is a public constant, and then
consider a parallel composition of $n$ such processes:
$P_n := \Pi_{i=1}^n R_i$. Thanks to compression, we will only
consider traces made of blocks, and  obtained a first exponential
reduction of the state space.  Now, assume
that our order prioritizes blocks on
$c_i$ w.r.t. those on $c_j$ when $i < j$, and consider a trace starting with
$\In(c_j, \cdot). \Out(c_j, w_j)$. Trying to continue the exploration
with an input action on $c_i$ with $i < j$, the dependency constraint added will impose
that the recipe $R_i$ used to feed the input on $c_i$ makes use of the previous
output to derive $\ok$. Actually, this dependency constraint will impose more than
that. Indeed, imposing that $R_i$ has to use $w_j$ is not very
restrictive since in this example there are many ways to rely on $w_j$
to derive $\ok$, \emph{e.g.} $R_i = \dec{\enc{\ok}{w_j}}{w_j}$ or even
  $R_i = w_j$.

Actually, according to our reduced semantics, 
this step will be possible only
if the message stored in $w_j$ is mandatory to derive  the message
expected in input on channel $c_i$.
In this example, the conditional imposes that the recipe $R_i$ 
leads to the message $\ok$ (at least to pursue in the then branch),
and even if there are many ways to derive
$\ok$, in any case the content of $w_j$ is not mandatory for that.
Thus, on this simple example, all the traces where the action
$\In(c_j,\cdot)$ is performed after the block on $c_j$ with $i < j$ will
not be explored thanks to our reduction technique.
\end{example}

The constraint solver is then modified in a non-invasive way:
dependency constraints are used to dismiss configurations when
it becomes obvious that they cannot be satisfied.
The modified verification algorithm may explore symbolic
traces that do not correspond to $\Phi$-minimal representatives
(when dependency constraints cannot be shown to be infeasible)
but we will see that this approach allows us to obtain a
very effective optimisation.
Note that, because we may over-approximate dependency constraints,
we must ensure that constraint resolution prunes executions
in a symmetrical fashion for both processes being checked for equivalence.

\begin{remark}
  A subtle point about compression is that it actually enhances reduction
  in a symbolic setting.
  Consider the process
  $P = \In(c,x).\Out(c,n_1).\Out(c,n_2)$
  in parallel with $Q = \In(c',x)$.
  If ${\prec}$ gives priority to $Q$,
  then $Q$ can only be scheduled after $P$ if its input message
  requires the knowledge of one of the nonces~$n_1$ and~$n_2$ revealed by $P$.
  Thus we have two symbolic interleavings, one of which is subject to
  a dependency constraint.
  Now, we could have applied the ideas of reduction directly on
  actions rather than blocks but we would have obtained
  three symbolic interleavings, reflecting the fact that
  if the input on $c'$ depends on only the first output nonce,
  it should be scheduled before the second output.
\end{remark}

\subsection{Experimental results}

The optimisations developed in the present paper have been implemented,
following the above approach, in the official version of the
state of the art tool
\Apte~\cite{apte-github}.
We now report on experimental results;
sources and instructions for reproduction are available~\cite{apte-por}.
We only show examples in which equivalence holds, because the time
spent on inequivalent processes is too sensitive to the order in which
the (depth-first) exploration is performed.

\renewcommand{\myparagraph}[1]{\medskip\noindent\emph{#1}~}

\myparagraph{Toy example.} We consider again our simple example
described in Section~\ref{subsec:symbolic}.
We ran \Apte on $P_n \approx P_n$ for $n = 1$ to $22$,
on a single 2.67GHz Xeon core (memory is not relevant).
We performed our tests on the reference version
and the versions optimised with the compressed and reduced semantics
respectively.
The results are shown on the left graph of Figure~\ref{fig:graphs},
in logarithmic scale: it confirms that
each optimisation brings an exponential speedup,
as predicted by our theoretical analysis.

\begin{figure}[htpb]
  \begin{center}
    \includegraphics{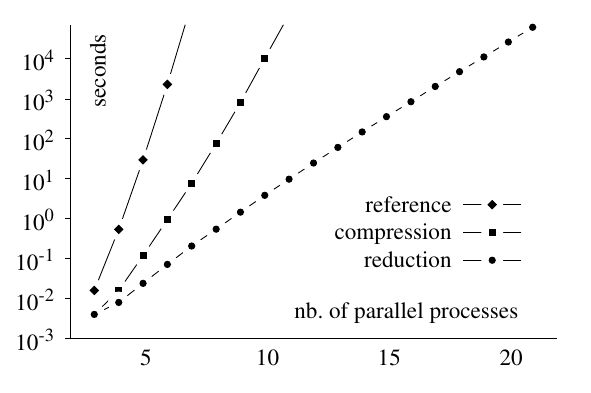}
    \includegraphics{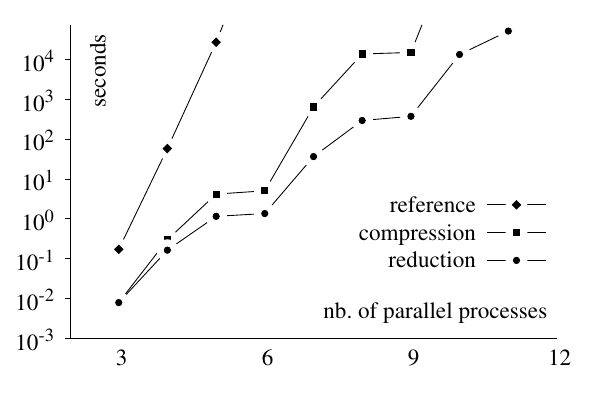}
    \vspace{-0.5cm}
  \end{center}
  \caption{Impact of optimisations on toy example (left)
  and Denning-Sacco (right).}
  \label{fig:graphs}
\end{figure}

\myparagraph{Denning-Sacco protocol.} We ran a similar benchmark,
checking that Denning-Sacco ensures strong secrecy in various scenarios. The protocol
has three roles and we added processes playing those roles in turn,
starting with three processes in parallel.
The results are plotted on 
Figure~\ref{fig:graphs}.
The fact that we add one role out of three at each step
explains the irregular growth in verification time. 
We still observe an exponential speedup for each optimisation.

\myparagraph{Practical impact.}
Finally, we illustrate how our optimisations make \Apte much more useful
in practice for investigating interesting scenarios.
Verifying a single session of a protocol brings little assurance
into its security. In order to detect replay attacks and to allow the
attacker to compare messages that are exchanged, at least two sessions
should be considered.
This means having at least four parallel processes for two-party protocols,
and six when a trusted third party is involved.
This is actually beyond 
what the unoptimised \Apte can handle in a reasonable amount
of time.
We show below
how many parallel processes could be handled in $20$ hours
by the different versions of \Apte on various use cases of protocols.

\medskip{}

\noindent\begin{tabular}{|l|>{\centering\arraybackslash}p{0.3cm}|>{\centering\arraybackslash}p{0.6cm}|>{\centering\arraybackslash}p{0.3cm}
||l|>{\centering\arraybackslash}p{0.3cm}|>{\centering\arraybackslash}p{0.6cm}|>{\centering\arraybackslash}p{0.3cm}|} \hline
 Protocol & ref & comp & red &Protocol & ref & comp & red \\ \hline
      Needham Schroeder (3-party)& 4 & 6  & 7 &   Denning-Sacco (3-party) & 5 & 9  & 10 \\
  Private Authent. (2-party)& 4 & 7 & 7 &   WMF (3-party)& 6 & 12 & 13 \\
  Yahalom (3-party) & 4 & 5 &5 &  E-Passport PA (2-party) & 4 & 7 & 9 \\ \hline

\end{tabular}

\section{Related Work}
\label{sec:relWork}

The techniques we have presented borrow from standard ideas
from concurrency theory, trace theory and, perhaps more surprisingly,
proof theory. Blending all these ingredients, and adapting them to the
demanding framework of security protocols, we have come up with
partial order reduction techniques that can effectively be used in
symbolic verification algorithms for equivalence properties
of security protocols.
We now discuss related work, and there is a lot of it given the
huge success of POR techniques in various application areas.
We shall focus on the novel aspects of our approach,
and explain why such techniques have not been needed outside
of security protocol analysis. These observations are not new:
as pointed out by Baier and Katoen~\cite{ModelCheckingBook},
``[POR] is mainly appropriate to control-intensive applications and less
suited for data-intensive applications'';
Clarke \emph{et al.}~\cite{clarke2000partial} also remark that
``In the domain of model checking of
reactive systems, there are numerous techniques for reducing
the state space of the system. One such technique
is partial-order reduction. This technique does not directly
apply to [security protocol analysis] 
because we explicitly keep
track of knowledge of various agents, and our logic can
refer to this knowledge in a meaningful way.''

We first compare our work with classical POR techniques.
Then, we discuss more specifically previous works that use POR
in a symbolic execution setting, and comment
on previous work in the domain of security protocol analysis.
We conclude the section with some remarks on the relationship
between our optimized semantics and focused proof systems.

%

\subsection{Classical POR}

Partial order reduction techniques
have proved very useful in the domain of model checking concurrent programs.
Given a Labelled Transition System (LTS) and some property to check
(\eg a Linear Temporal Logic formula), 
the basic idea of POR~\cite{Peled98,godefroid1996partial,ModelCheckingBook}
is to only consider a reduced
version of the given LTS whose enable transitions of some states might be not exhaustive
but are such that this transformation does not affect the property.
POR techniques can be categorized in two groups~\cite{godefroid1996partial}.
First, the {\em persistent set} techniques
(\eg {\em stubborn sets, ample sets}) where only a
sufficiently representative subset of available transitions is explored.
Second, {\em sleep set} techniques memoize past exploration
and use this information along with available transitions to
disable some provably redundant transitions.
Note that these two kinds of techniques are compatible, and are indeed
often combined to obtain better reductions.
Theoretical POR techniques apply to transition systems which may not be
explicitly available in practice, or whose explicit computation may be
too costly. In such cases, POR is often applied to an approximation of
the LTS that is obtained through static analysis. Another, more recent
approach is to use
\emph{dynamic 
POR}~\cite{flanagan2005dynamic,tasharofi2012transdpor,abdulla2014optimal}
where the POR arguments are applied
based on information that is obtained during the execution of the
system.


Clearly, classical POR techniques would apply
to our concrete LTS, but that would not be practically useful since this
LTS is wildly infinite, taking into account all recipes that the attacker
could build.
Applying most classical POR techniques to the LTS from which
data would have been abstracted away would be ineffective:
any input would be dependent with any output (since the attacker's
knowledge, increased by the output, may enable new input messages).
Our compression technique lies between these two extremes. It exploits
a semi-commutation property: outputs can be permuted before inputs,
but not the converse in general.
Further, it exploits the fact that inputs
do not increase the attacker's knowledge, and can thus be executed
in a chained fashion, under focus.
The semi-commutation is reminiscent of the asymmetrical dependency
analysis enabled by the \emph{conditional} stubborn set
technique~\cite{godefroid1996partial},
and the execution of inputs under focus may be explained
by means of sleep sets.
While it may be possible to formally derive our compressed semantics
by instantiating abstract POR techniques to our setting,
we have not explored this possibility in detail\footnote{
  Although this would be an interesting question, we do not expect
  that any improvement of compression would come out of it.
  Indeed, compression can be argued to be maximal in terms of
  eliminating redundant traces without analyzing data:
  for any compressed trace there is a way to choose messages
  and modify tests to obtain a concrete execution
  which does not belong to the equivalence class
  of any other compressed trace.
}.
As mentioned earlier, the compressed semantics is inspired from
another technique, namely focusing~\cite{Andreoli92} from proof theory.
Concerning
our reduced semantics, it may be seen as an application of the sleep set
technique (or even as a reformulation of Anisimov's and Knuth's
characterization of lexicographic normal forms) but the real contribution
with this technique is to have formulated it in such a way that it can
be implemented without requiring an \emph{a priori} knowledge of
data dependencies: it allows us to eliminate redundant traces on-the-fly
as data (in)dependency is discovered by the constraint resolution procedure
(more on this in the next sections)
--- in this sense, it may be viewed as a case of dynamic POR.

Narrowing the discussion a bit more, we now focus on the fact that
our techniques are designed for the verification of equivalence properties.
This requirement turns several seemingly trivial observations into
subtle technical problems. For instance,
ideas akin to compression are often applied without justification
(\eg in~\cite{sen2006automated,tasharofi2012transdpor,ModersheimVB10})
because they are obvious when one does
reachability rather than equivalence checking.
To understand this,
it is important to distinguish between two very
different ways of applying POR to equivalence checking
(independently of the precise equivalence under consideration).
The first approach is to reduce a system such that the reduced system
and the original systems are equivalent.
In the second approach, one only requires that two reduced
systems are equivalent iff the original systems are equivalent.
The first approach seems to be more common in the POR 
litterature (where one finds, \eg reductions that preserve
LTL-satisfiability~\cite{ModelCheckingBook}
or bisimilarity~\cite{huhn98fsttcs})
though there are instances of the second approach
(\eg for Petri nets~\cite{godefroid1991using}).
In the present work, we follow the second approach:
neither of our two reduction techniques preserves trace equivalence.
This allows stronger reductions but requires extra care:
one has to ensure that the independencies used in the reduction of one
process are also meaningful for the other processes; in other words,
reduction has to be symmetrical.
This is the purpose of our annotated semantics and its ``strong symmetry
lemma'' (Lemma~\ref{lem:strong-symmetry}) but also, for the reduced semantics,
of Lemma~\ref{lem:min-swap} and Proposition~\ref{prop:eqtb-sym}.
We come back to these two different approaches later, when discussing
specific POR techniques for security.

\subsection{Infinite data and symbolic execution}

Symbolic execution is often used to verify systems dealing with infinite data,
\eg recipes and messages in security; integers, list, etc. in program
verification.
In many works combining POR and symbolic executions
(\eg \cite{sen2006automated,siegel2006using,tasharofi2012transdpor})
the detection of redundant explorations does not rely on the data,
and can thus be done trivially at the level of symbolic executions.
In such cases, POR and symbolic execution are orthogonal.
For instance, in~\cite{sen2006automated},
two actions are \emph{data-dependent}
if one actions is a {\em send} action of some message $m$ to some process $p$ 
and the other is a {\em receive} action of process $p$.
This is done independently of $m$, which is meaningful when one considers
internal reduction (as is the case that work) but would be too coarse when one
considers labelled transitions representing interactions with an environment
that may construct arbitrary recipes from previous outputs (as is the case
in our work). Due to this omnipotent attacker, POR techniques cannot be
effective in our setting unless they really take data into account. Further,
due to the infinite nature of data, and the dynamic extension of the
attacker's knowledge, POR and symbolic execution must be integrated
rather than orthogonal.
Our notion of independence
and trace monoid are tailored to take all this into account,
including at the level of trace equivalence
(cf.~again Lemma~\ref{lem:min-swap}).
And, although reduction is presented (Section~\ref{sec:reduction})
in the concrete semantics, the crucial point is that
the {\em authorised} predicate (Definition~\ref{def:dep-analysis})
is implemented as a special \emph{dependency} constraint
(Section~\ref{sec:application}) so that our POR algorithm fundamentally
relies on symbolic execution.
We have not seen such uses of POR outside of the security applications
mentioned next~\cite{ModersheimVB10,BDH-post14}.

\subsection{Security applications}

The idea of applying POR to the verification of security protocols dates
back, at least, to the work of
Clarke \emph{et al.}~\cite{clarke2000partial,ClarkeJM03}.
In this work, the authors remark that traditional POR techniques
cannot be directly applied to security mainly because
``[they] must keep track of knowledge of various agents''
and ``[their] logic can refer to this knowledge in a meaningful way''.
This led them to define a notion of {\em semi-invisible actions}
(output actions, that cannot be swapped after inputs but only before them)
and design a reduction that prioritizes outputs
and performs them in a fixed order.
Compared to our work, this reduction is much weaker
(even weaker than compression only),
only handles a finite set of messages, and
only focuses on reachability properties checking.

In \cite{escobar14ic}, the authors develop ``state space reduction''
techniques for the Maude-NRL  Protocol Analyzer (Maude-NPA). This tool proceeds
by backwards reachability analysis and treats at the same level
the exploration of protocol executions and attacker's deductions.
Several reductions techniques are specific to this setting,
and most are unrelated to partial order reduction in general,
and to our work in particular.
We note that the lazy intruder techniques from \cite{escobar14ic}
should be compared to what is done in constraint resolution procedures
(\eg the one used in \Apte) rather than to our work.
A simple POR technique used in Maude-NPA is based on the observation
that inputs can be executed in priority in the backwards exploration,
which corresponds to the fact that we can execute outputs first in
forward explorations. We note again that this is only one aspect
of the focused strategy, and that it is not trivial to lift this observation
from reachability to trace equivalence.
Finally, a ``transition subsumption'' technique is described for Maude-NPA.
While highly non-trivial due to the technicalities of the model,
this is essentially a tabling technique rather than a partial order
reduction. Though it does yield a significant state space reduction
(as shown in the experiments \cite{escobar14ic}) it falls short of
exploiting independencies fully, and has a potentially high computational
cost (which is not evaluated in the benchmarks of 
\cite{escobar14ic}).

In \cite{fokkink2010partial},
Fokkink {\it et al.} model security protocols as labeled transition
systems whose states contain the control points of different agents
as well as previously outputted messages.
They devise some POR technique for these transition systems,
where output actions are prioritized and performed in a fixed order.
In their work, the original and reduced systems
are trace equivalent {\em modulo} outputs
(the same traces can be found after removing output actions).
The justification for their reduction would fail in our setting,
where we consider standard trace equivalence with observable outputs.
More importantly, their requirement that a reduced system should be
equivalent to the original one makes it impossible to swap input actions,
and thus reductions such as the execution under focus of our
compressed semantics cannot be used.
The authors leave as future work the problem of combining
their algorithm with symbolic executions, in order to be able
to lift the restriction to a finite number of messages.

Cremers and Mauw proposed~\cite{cremers2005checking} a reduction technique
based on the {\em sleep set} idea. Basically, when their exploration algorithm
chooses to explore a specific action (an output or an input with its corresponding message),
it will also add all the other available actions
that have priority over the chosen one to the current {\em sleep set}.
An action in this sleep set will never been explored.
In this work, only reachability property are considered, and the
reduction cannot be directly applied to trace equivalence checking.
More importantly, the technique can only handle a finite set of messages.
The authors identify as important future work the
need to lift their method to the symbolic setting.

Earlier work by Mödersheim {\it et al.}
has shown how to combine POR technique with symbolic semantics~\cite{ModersheimVB10}
in the context of reachability properties for security protocols,
which has led to high efficiency gains in the OFMC tool of the
AVISPA platform~\cite{sysdesc-CAV05}.
While their reduction is very
limited, it brings some key insight on how POR may be combined
with symbolic execution.
In a model where actions are sequences of outputs followed by inputs,
their reduction imposes a \emph{differentiation} constraint on the
interleavings of $\In(c,x).\Out(c,m)\Par\In(d,y).\Out(d,m')$.
This constraint enforces that the symbolic interleaving
${\In(d,M').\Out(d,w').\In(c,M).\Out(c,w)}$
should only be explored for instances of $M$ that depend on $w'$.
Our reduced semantics constrains patterns of arbitrary size (instead of
just size 2 diamond patterns as above) by means of the
{\em authorised} predicate (Definition~\ref{def:dep-analysis}).
Moreover, our POR technique has beed designed to be sound and complete
for trace equivalence checking as well.

In a previous work \cite{BDH-post14},
we settled the general ideas for the POR techniques presented in the
present paper, but results were much weaker.
That earlier work only dealt with the restrictive class of
\emph{simple processes}, which does not feature nested parallel
composition or replication, and made heavy use of specific properties
of processes of that class to define reductions and prove them correct.
In the present work, we show that our two reduction techniques apply
to a very large class of processes for reachability checking.
For equivalence checking, we only require the semantic condition of
action-determinism.
Note that the results of the present paper are conservative over those
of \cite{BDH-post14}: the reductions of \cite{BDH-post14} are obtained
as a particular case of the results presented here in the case of
simple processes.
Finally, the present work brings a solid implementation in the state of the 
art tool Apte~\cite{Cheval-tacas14},
whereas \cite{BDH-post14} did not present experimental results ---
it mentioned a
preliminary implementation, extending SPEC, which we abandoned since it
was difficult to justify and handled a more restricted class
than APTE.


\subsection{Relationship with focused proof systems}

The reader familiar with focused proof systems \cite{Andreoli92} will have
recognized the strong similarities with our compressed semantics.
The strategies are structured in the same way, around positive and negative
phases. More deeply, the compressed semantics can actually be derived
systematically from the focused proof system of linear logic, through an
encoding of our processes into linear logic formulas (such that proof
search corresponds to process executions).
There are several such encodings in the
literature, see for instance \cite{Miller03,GargP05,Simmons12}.
%
We do not provide here a fully worked-out encoding appropriate for
our protocols.
It is not trivial, notably due to the need to encode the attacker's
knowledge, and internal reductions of protocols
--- both features require slight extensions of the usual linear logic
framework.
We have thus chosen to only take the correspondence with linear logic
as an intuitive guide, and give a self-contained (and simple) proof
of completeness for our compressed semantics
by adapting the positive trunk argument of~\cite{saurin}. 
Note that the strong analogies with proof theory only hold for
reachability results, \ie up to Lemma~\ref{thm:reach-comp}.
It is a contribution of this paper to observe that
focusing (compression) makes sense beyond reachability,
at the level of trace equivalence:
Theorem~\ref{thm:eintc} (stating that trace equivalence
coincides with compressed trace equivalence for action-deterministic
processes) has no analogue in the proof theoretical setting, where
trace equivalence itself is meaningless.

We motivated reduction by observing that (in)dependencies between
blocks of the compressed semantics should be exploited to eliminate
redundant interleavings. This same observation has been done in the
context of linear logic focusing, and lead to the idea of
multi-focusing~\cite{chaudhuri08tcs} where independent synthetic
connectives (the analogue of our blocks) are executed simultaneously
as much as possible. That work on multi-focusing is purely theoretical,
and it is unclear how multi-focusing could be applied effectively
in proof search. It would be interesting to consider whether
the gradual construction of unique representatives in our
reduced semantics could be extended to the richer setting of
linear logic (where proof search branches, unlike process executions).

\section{Conclusion} \label{sec:conclusion}

We have developed two POR techniques that are adequate for verifying
reachability and trace equivalence properties of action-deterministic
security protocols. We have effectively implemented them in \Apte,
and shown that they yield the expected, significant benefit.

We are considering several directions for future work.
Regarding the theoretical results presented here, the main question is whether we can get 
rid of the action-determinism condition without degrading our reductions too 
much.
Regarding the practical application of our results, we can certainly
go further.
We first note that our compression technique should be applicable
and useful in other verification tools,
not necessarily based on symbolic execution.
 Next, we could investigate the role of the particular choice
 of the order $\prec$, to determine heuristics for maximising the
 practical impact of reduction.
Finally, we plan to adapt our treatment of replication to bounded
replication to obtain a first symmetry elimination scheme, which should
provide a significant optimisation when studying security protocols
with several sessions.

\bibliography{biblio} 

\newpage

\appendix

\section{Annotated semantics}
\label{sec:app:prelim}

\subsection{Reachability}

\begin{proposition}
  Well labelling is preserved by $\sinta{}$.
  \label{prop:wlabel-stable}
\end{proposition}

\begin{proof}
  For all transitions except {\sc Par} and {\sc Repl}, the multiset
  of labels of the resulting configuration is a subset of labels
  of the original configuration. Thus, well labelling is obviously
  preserved in those cases.

  Consider now a {\sc Par} transition, represented below without
  the permutation, which does not play a role here:
  \[
  \proc{\loc{\{\Pi_{i=1}^n P_i\}}{\ell}\}\uplus\p}{\Phi}
  \lrstep{\loc{\Para(\sigma_{1}, \ldots, \sigma_{n})}{\ell}}
  \proc{\{\loc{P_{1}}{\ell \cdot 1},\ldots,\loc{P_{n}}{\ell \cdot n}\}\uplus\p}{\Phi}
  \]
  We check that labels are pairwise independent.
  This is obviously the case for the new labels $\ell \cdot i$.
  Let us now consider a label $\ell'$ from $\p$ and show that it is
  independent from any $\ell \cdot i$. It cannot be equal to one of them
  (otherwise it would be a suffix of $\ell$, which contradicts the
  well labelling of the initial configuration) and it cannot be
  a strict prefix either (otherwise it would be a prefix of $\ell$ too).
  Finally, $\ell \cdot i$ cannot be a prefix of $\ell'$ because $\ell$
  is not a prefix of $\ell'$.

  As far as labels are concerned, {\sc Repl} transitions are a particular
  case of {\sc Par} where there are only two sub-processes. Thus {\sc Repl} 
  preserves well labelling.
\end{proof}


\restatelemma{lem:perm}

\begin{proof}
  By symmetry it is sufficient to show one implication. Assuming $\ell_1$
  and $\ell_2$ to be independent, we consider a transition labelled
  $\alpha=\loc{\alpha'}{\ell_1}$ followed by one labelled $\beta=\loc{\beta'}{\ell_2}$.
  We first observe that a transition labelled $\ell_1$ can only generate
  new labels that are dependent with $\ell_1$. Thus, $\ell_2$ must
  be present in the original configuration and our execution is of the
  following form, where we write
  $P_\alpha$ (resp.~$P_\beta$) instead of $\loc{P_\alpha}{\ell_1}$
  (resp.~$\loc{P_\beta}{\ell_2}$):
  \[
  A = \proc{\mathcal{P}\uplus \{P_\alpha,P_\beta\}}{\Phi}\sint{\loc{\alpha}{\ell_1}}
  \proc{\p \uplus \p_\alpha \uplus \{P_\beta\}}{\Phi_\alpha}\sint{\loc{\beta}{\ell_2}}
  \proc{\p\uplus \p_\alpha \uplus \p_\beta}{\Phi_\beta}
  \]
  It remains to check that $\beta$ can be performed by $P_\beta$ in the
  original configuration,
  and that doing so would not prevent the $\alpha$ transition to happen
  next. The only thing that could prevent $\beta$ from being performed
  is that the frames $\Phi$ and $\Phi_\alpha$ may be different, in the
  case where $\alpha$ is an input. In that case, the recipe
  independence hypothesis guarantees that $\beta$ does not rely on the
  new handle introduced by $\alpha$ and can thus be played with only
  $\Phi$.
  Finally, performing $\alpha$ after $\beta$ is easy.
  We only detail the case where
  $\beta=\Out(c,w)$ and $\alpha$ is an input of recipe $M$.
  In that case we have
  $\Phi_\alpha = \Phi$, $\Phi_\beta=\Phi_\alpha\uplus\{w\mapsto m\}$,
  and $M \in \T(\dom(\Phi))$.
  We observe that $M \in \T(\dom(\Phi_\beta))$ and we construct the
   execution:
\[
A= \proc{ \mathcal{P}\uplus \{\loc{P_\alpha}{\ell_1},\loc{P_\beta}{\ell_2}\}}{\Phi}\sint{\loc{\beta}{\ell_2}}
  \proc{\mathcal{P}\uplus \p_\beta \uplus \{\loc{P_\alpha}{\ell_1}\}}{\Phi_\beta}\sint{\loc{\alpha}{\ell_1}}
  \proc{\mathcal{P}\uplus \p_\alpha \uplus \p_\beta}{\Phi_\beta}
\]
\end{proof}

\subsection{Equivalence}
\begin{definition}
  Given a process $P$, we define the set of its {\em enabled skeletons} as
  \[\enab(P)=
  \begin{cases}
    \{\sk(P)\} & \text{if $P$ starts with an observable action}\\
    \cup_i\{\sk(P_i)\} & \text{if $P=\Pi_i P_i$}\\
    \emptyset & \text{if $P=0$}\\
  \end{cases}\]
\end{definition}
We may consider skeletons, labelled skeletons and enabled skeletons of a 
configuration
by taking the set of the corresponding objects of all its processes.
\begin{property}
  For any configurations $A$, $A'$ and non-observable action $\alpha$,
  if $A\sinta{\alpha}A'$ or $A\sint{\tau}A'$ then $\enab(A)=\enab(A')$.
\label{prop:enab-stable}
\end{property}
\begin{property}
 Let $A$ be an action-deterministic configuration and $P$, $Q$ two of its processes.
 We have that $\enab(P)\cap\enab(Q)=\emptyset$.
 \label{prop:det-enab}
\end{property}

\begin{lemma}
  Let $A$ be an action-deterministic configuration.
  If $A\sint{\tr_1}A_1$ and $A\sint{\tr_2}A_2$ for some traces $\tr_1,\tr_2$ such that
  $\obs(\tr_1)=\obs(\tr_2)$ then $\enab(A_1)=\enab(A_2)$ and  $\Phi(A_1) = \Phi(A_2)$.
\label{lem:action-det-enab}
\end{lemma}

\begin{proof}
We first prove a stronger result when the configurations $A_1$ and
$A_2$ are \emph{canonical}, \emph{i.e.} only contain processes that are neither $0$ nor a parallel
composition. Actually, in such a case, we prove that $A_1 = A_2$.

\medskip{}

To prove this intermediate result, we proceed by induction on $\obs(\tr_1)$. The base case is trivial.
Let us show the inductive case.
We assume that $\tr_1=\tr_1^0.\alpha.\tr^-_1$ with $\alpha$ an observable
action and $\tr_1^-$ containing only non-observable
actions.
Since $\obs(\tr_1)=\obs(\tr_2)$, we have that
$\tr_2=\tr_2^0.\alpha.\tr^-_2$ with $\tr_2^-$ containing only non-observable
actions and $\obs(\tr_1^0)=\obs(\tr_2^0)$.
Our given executions are thus of the form:
\[
A\sint{\tr_1^0}A^0_1\sint{\alpha.\tr_1^-}A_1 \text{ and }
A\sint{\tr_2^0}A^0_2\sint{\alpha.\tr_2^-}A_2
\]
It may be the case that $A_1^0$ or $A_2^0$ are not canonical.
The idea is to reorder some non-observable actions.
More precisely, we perform all available non-observable
actions of $A^0_1$ and $A^0_2$ before performing $\alpha$.
By doing this, we do not change the observable actions of the different sub-traces and obtain
\[
A\sint{\tr_1^0.\tr_3^-}A'^0_1\sint{\alpha.\tr'^-_1}A_1 \text{ and }
A\sint{\tr_2^0.\tr_4^-}A'^0_2\sint{\alpha.\tr'^-_2}A_2
\]
with $A'^0_1$ and $A'^0_2$ canonical.
By inductive hypothesis, we have that
$A'^0_1 = A'^0_2$.
We now must show $A_1=A_2$. By action-determinism of $A$, there is only one process $P$
that can perform $\alpha$ in $A'^0_1(=A'^0_2)$. The resulting process $P'$ after performing $\alpha$
is thus the same in the two executions. Since $A_1$ and $A_2$ are canonical and
$\tr'^-_1$ and $\tr'^-_2$ contain only non-observable actions, $A_1=A_2$.

\medskip{}

In order to be able to apply our previous result, 
we  complete the executions with all available non-observable actions:
\begin{center}
$A\sint{\tr_1}A_1\sint{\tr_1^-}A_1'$ and
$A\sint{\tr_2}A_2\sint{\tr_1^-}A_2'$ 
\end{center}
such that $A_1$ and $A_2$ are canonical
and $\tr^-_1$ and $\tr^-_2$ contain only non-observable actions.
We also have that:
\begin{itemize}
\item $\Phi(A_1)=\Phi(A'_1)$ and $\enab(A_1)=\enab(A'_1)$; and
\item  $\Phi(A_2)=\Phi(A'_2)$ and $\enab(A_2)=\enab(A'_2)$.
\end{itemize}
We now conclude thanks to our previous result, and obtain $A'_1=A'_2$ implying
the desired equalities.
\end{proof}

\begin{proposition}
\label{prop:det-eint}  
Let $A$ and $B$ be two action-deterministic configurations such that $A\eint B$.
If $A\sint{\tr_A}A'$ and $B\sint{\tr_B}B'$ with $\obs(\tr_A) =
\obs(\tr_B)$  then $\Phi(A')\estat\Phi(B')$ and
$\enab(A')=\enab(B')$.
\end{proposition}

\begin{proof}
By hypothesis, we know that $A\eint B$, and also that
$A\sint{\tr_A}A'$. 
Moreover, the freshness conditions on channels
(\ie $\bc(\tr_A)\cap\fc(B)=\emptyset$)
holds as $B$ is able to perform $\tr_B$,
and $\tr_A$ and $\tr_B$ share the same bound channels.
Hence, we know that there exist $\tr'_B$ and $B''$
such that
\begin{center}
  $B \sint{\tr'_B} B''$, $\obs(\tr_A) = \obs(\tr'_B)$, and
  $\Phi(A') \estat \Phi(B'')$. 
\end{center}
Now, since $B$ is an action-deterministic configuration, applying
Lemma~\ref{lem:action-det-enab} on $\tr_B$ and $\tr'_B$,
we obtain that $\enab(B') = \enab(B'')$ and
$\Phi(B') = \Phi(B'')$. This allows us to conclude that
$\Phi(A')\estat \Phi(B')$.

It remains to show that $\enab(A') = \enab(B')$.
By symmetry, we only show one inclusion.
Let $\alpha_s \in \enab(A')$, we shall show that $\alpha_s\in\enab(B')$.
We deduce from the latter that there is a trace
$\tr'$ that is either $\alpha$ or $\tau.\alpha$ (where $\alpha$
is an observable action whose the skeleton is $\alpha_s$)
such that
$$A \sint{\tr_A} A'\sint{\tr'}A_0$$
for some $A_0$.
Since $A \eint B$, we know that there
exist $\tr'_0$, $\tr'$,  $B'_0$, and $B'$ such that
$$B \sint{\tr'_0} B'_0 \sint{\tr'} B_0$$
with $\Phi(A_0) \estat \Phi(B_0)$ and $\obs(\tr_A)= \obs(\tr_0')$
and $\obs(\tr')=\alpha$.
we have that $\alpha \in \enab(B'_0)$.
In particular, using Property~\ref{prop:enab-stable},
we have that $\alpha_s\in\enab(B'_0)$.

Now, since $B$ is an action-deterministic configuration, applying
Lemma~\ref{lem:action-det-enab} on $\tr_B$ and $\tr'_0$
we obtain $\enab(B') =
  \enab(B'_0)$, and thus $\alpha_s \in \enab(B')$.
\end{proof}

\begin{proposition} \label{prop:skfromo}
  Let $A$ be an action-deterministic configuration and $P,Q$ two
  of its processes.
  If $\enab(P)=\enab(Q)$ then $\sk(P)=\sk(Q)$.
\end{proposition}

\begin{proof}
Let us show that $\sk(P) = \sk(Q)$.
If $\enab(P)$ is an empty set then $P = 0$
and thus from $\enab(Q)=\emptyset$ we deduce that
$Q=0$ as well implying the required equality on skeletons.
If $\enab(P)$ is a singleton then it must be $\{\sk(P)\}$ ---
we cannot be in the case where $P$ is a parallel composition,
for in that case there would be at least two skeletons in $\enab(P)$
by action-determinism of $A$.
The same goes with $Q$ thus we have $\{\sk(P)\}=\{\sk(Q)\}$.
Finally, if $\enab(P)$ contains at least two skeletons
then it must be the case that $P$ is a parallel composition
of the form $\Pi_iP_i$ and $\enab(P) = \cup_i\{\sk(P_i)\}$.
Similarly, $Q$ must be of the form $\Pi_i Q_i$ and 
$\enab(Q) = \cup_i\{\sk(Q_i)\}$.
Here, we make use of action-determinism to obtain that the number of
subprocesses in parallel is the same as the cardinality of the sets of
skeletons, and thus the same for $P$ and $Q$:
indeed, no two parallel subprocesses can have the same skeleton.
We conclude that $\sk(P)=\Para(S)$ where
$S$ is the ordered sequence of skeletons from $\cup_i\{\sk(P_i)\}$,
and $\sk(Q)=\Para(S)$ where $S$ is the ordered sequence of skeletons
from $\cup_i\{\sk(Q_i)\} = \cup_i\{\sk(P_i)\}$.
\end{proof}


\restatelemma{lem:strong-symmetry}

\begin{proof}
We show this result by induction on the length of the derivation $A
\lrstep{\tr} A_n$. The
case where $\tr$ is empty (\ie no action even a non-observable one) is
obvious. Assume that we have proved such a result for all the executions of
length~$n$, and we want to establish the result for an execution of
length~$n+1$. 

Consider an execution of 
$\loc{\alpha_1}{\ell_1}\ldots\loc{\alpha_n}{\ell_n}$
from $A$ to $A_n$, followed by $\loc{\alpha_{n+1}}{\ell_{n+1}}$
towards $A_{n+1}$.
By induction hypothesis, we know that there exists an execution
$$B \lrstep{\loc{\alpha_1}{\ell_1}} B_1
\lrstep{\loc{\alpha_2}{\ell_2}}  \ldots
\lrstep{\loc{\alpha_n}{\ell_n}} B_n$$
such that
$\Phi(A_n) \estat \Phi(B_n)$ and
$\skl(A_i) = \skl(B_i)$ for any $1 \leq i \leq n$.
It remains to establish that there exists $B_{n+1}$ such that
$B_n$ can perform $\loc{\alpha_{n+1}}{\ell_{n+1}}$ towards
$B_{n+1}$, $\Phi(A_{n+1}) \estat \Phi(B_{n+1})$ and
$\skl(A_{n+1}) = \skl(B_{n+1})$.
We distinguish several cases depending on the action $\alpha_{n+1}$.
\bigskip{}

\noindent \emph{Case $\alpha_{n+1} = \Zero$.}
We have that $\loc{\Zero}{\ell_{n+1}} \in \skl(A_n)$, and thus,
since $\skl(A_n) = \skl(B_n)$, we
have also that
$\loc{\Zero}{\ell_{n+1}} \in \sk(B_n)$. 
We deduce that
 $A_n = \proc{\{\loc{0}{\ell_{n+1}}\} \uplus \p_0}{\Phi_0}$, and
$B_n = \proc{\{\loc{0}{\ell_{n+1}}\} \uplus \q_0}{\Psi_0}$
for some $\p_0$, $\q_0$, $\Phi_0$, and $\Psi_0$. Moreover, since
$\skl(A_n) =\skl(B_n)$, we deduce that 
$\skl(\p_0) = \skl(\q_0)$. 
Let $B_{n+1}  =
  \proc{\q_0}{\Psi_0}$. We have that:
\begin{itemize}
\item  $B_n = \proc{\{\loc{0}{\ell_{n+1}}\} \uplus \q_0}{\Psi_0}
  \lrstep{\loc{\Zero}{\ell_{n+1}}} \proc{\q_0}{\Psi_0} = B_{n+1}$,
\item $\Phi(A_{n+1}) = \Phi(A_n) \estat \Phi(B_n) = \Phi(B_{n+1})$, and 
\item $\skl(A_{n+1}) =
  \skl(\p_0) = \skl(\q_0) = \skl(B_{n+1})$. 
\end{itemize}

\medskip{}

\noindent \emph{Case $\alpha_{n+1} = \Para(S)$
  for some sequence $S =
  (\beta_1, \ldots, \beta_k)$.}
Note that this sequence is ordered according to our
order $\ordS$ over skeletons
(\ie $\beta_1\ordS\ldots\ordS\beta_k$)
and $\beta_i$'s are pairwise distinct by action-determinism of $A$.
We have $\loc{\Para(S)}{\ell_{n+1}} \in \skl(A_n)$, and thus,
since $\skl(A_n) = \skl(B_n)$, we have also that
$\loc{\Para(S)}{\ell_{n+1}} \in \skl(B_n)$.
$A_n = \proc{\{\loc{\Pi_{i=1}^{k} P_i}{\ell_{n+1}} \} \uplus \p_0}{\Phi_0}$,
$\uplus_{i=1}^k \sk(P_i)=\{\beta_1,\ldots, \beta_k\}$ and
$B_n = \proc{\{\loc{\Pi_{i=1}^{k} Q_i}{\ell_{n+1}}\} \uplus \q_0}{\Psi_0}$
for some $P_i$, $Q_i$, $\p_0$, $\q_0$, $\Phi_0$, and
$\Psi_0$.
Further, we have
$$A_{n+1} =
\proc{\uplus_{i=1}^{k}\{\loc{P_{\pi(i)}}{\ell_{n+1}i}\} \uplus \p_0}{\Phi_0}$$
for some permutation $\pi$ over $[1;k]$ such that
$\sk(P_{\pi(i)})=\beta_{i}$ for all $i$.
Moreover, since $\skl(A_n)  = \skl(B_n)$, we deduce
that $\skl(\p_0) = \skl(\q_0)$, and 
\begin{center}
$\{\sk(P_i) ~|~ 1 \leq i \leq k\}  =
\{\sk(Q_i) ~|~ 1 \leq i \leq k \} = \{\beta_1, \ldots, \beta_k\}$
\end{center}
Remark that, since $P_i$ (resp. $Q_i$) cannot be a zero or a parallel
we have that $\enab(P_i)=\{\sk(P_i)\}$ (resp. $\enab(Q_i)=\{\sk(Q_i)\}$)
and those sets are singletons.
Moreover, by action-determinism of $A$ and $B$ we know that all those singletons are
pairwise disjoint.
From this,  we conclude that there exists a permutation $\pi'$ over $[1;k]$
such that
$$
 \forall i,\ \sk(Q_{\pi'(i)})=\beta_{i}=\sk(P_{\pi(i)})
$$
and thus
$$
\forall i,\ 
\skl(\loc{Q_{\pi'(i)}}{\ell_{n+1}i}) =
\skl(\loc{P_{\pi(i)}}{\ell_{n+1}i})
$$
We can finally let $B_{n+1}$ be
$\proc{\uplus_{i=1}^{k}\{\loc{Q_{\pi'(i)}}{\ell_{n+1}i}\}\uplus\q_0}{\Psi_0}$
and we have:
\begin{itemize}
\item $B_n = \proc{\{\loc{\Pi_{i=1}^{k} Q_i}{\ell_{n+1}}\} \uplus
    \q_0}{\Psi_0} \lrstep{\loc{\Para(S)}{\ell_{n+1}}} 
\proc{\uplus_{i=1}^{k}\{\loc{Q_{\pi'(i)}}{\ell_{n+1}i}\} \uplus \q_0}{\Psi_0} = B_{n+1}$, 
\item $\Phi(A_{n+1}) = \Phi(A_n) \estat \Phi(B_n) = \Phi(B_{n+1})$, and
\item
 $\skl(A_{n+1}) = \skl(\p_0) \; \uplus\; 
  \biguplus_{i=1}^k\skl(\loc{P_{\pi(i)}}{\ell_{n+1}i})$\\
\phantom{$\skl(A_{n+1})$} $=  \skl(\q_0) \;\uplus\; 
  \biguplus_{i=1}^k\skl(\loc{Q_{\pi'(i)}}{\ell_{n+1}i}) = \skl(B_{n+1})$.
\end{itemize}
\medskip{}
 
\noindent \emph{Case $\alpha_{n+1} = \In(c,M)$ for some $c$, and $M$
  with $M \in \T(\dom(\Phi(A_n))$.}
We have that $\loc{\In_c}{\ell_{n+1}} \in \skl(A_n)$, and thus,
since $\skl(A_n) = \skl(B_n)$, we have
$\loc{\In_c}{\ell_{n+1}} \in \skl(B_n)$.
We deduce that
$A_n = \proc{\{\loc{\In(c,x_A).P}{\ell_{n+1}}\} \uplus \p_0}{\Phi_0}$, and
$B_n = \proc{\{\loc{\In(c,x_B).Q}{\ell_{n+1}}\} \uplus \q_0}{\Psi_0}$
for some $x_A$, $x_B$, $P$, $Q$, $\p_0$, $\q_0$, $\Phi_0$, and
$\Psi_0$. Since $A \approx B$ and thus $\dom(\Phi(A))=\dom(\Phi(B))$,
we have that $\dom(\Phi_0) = \dom(\Psi_0)$.
Moreover, since $\skl(A_n) = \skl(B_n)$, we deduce that 
$\skl(\p_0) = \skl(\q_0)$.
Let
$B_{n+1} = \proc{\loc{Q\{{u_B}/{x_B}\}}{\ell_{n+1}} \uplus \q_0}{\Psi_0}$
where $u_B = M\Psi_0$.
We have that:
\begin{itemize}
\item $B_n \lrstep{\loc{\In(c,M)}{\ell_{n+1}}} B_{n+1}$, and 
\item $\Phi(A_{n+1}) = \Phi(A_n) \estat \Phi(B_n) = \Phi(B_{n+1})$. 
\end{itemize}
It remains to show
that $\skl(A_{n+1}) = \skl(B_{n+1})$. Since we have that
$\skl(\p_0) = \skl(\q_0)$, and the label of the new subprocess
is the same (namely, $\ell_{n+1}$) on both sides, we only need to show 
that:
\[
\sk(P\{M\Phi_0/x_A\}) =
\sk(Q\{M\Psi_0/x_B\})
\]

In order to improve the readability, we will note
$P'=P\{M\Phi_0/x_A\}$ and $Q'=Q\{M\Psi_0/x_B\}$.
We have that $A$ and $B$ are two action-deterministic configurations
such that $A \approx B$. Moreover, they perform the same trace, respectively
towards $A_{n+1}$ and $B_{n+1}$.
Thus, thanks to Proposition~\ref{prop:det-eint}, we deduce that
$\enab(A_{n+1}) = \enab(B_{n+1})$.
Moreover, our hypothesis $\skl(\p_0) = \skl(\q_0)$
implies that $\enab(\p_0) = \enab(\q_0)$, and thus we deduce
that $\enab(P') =\enab(Q')$ (recall that by action-determinism,
unions of the form $\enab(A_{n+1})=\enab(\p_0)\cup\enab(P')$ 
are actually disjoint unions).
We conclude using Proposition~\ref{prop:skfromo}.

\medskip{}

\noindent \emph{Case $\alpha_{n+1} = \Out(c,w)$ for some $c$ and some
  $w$ with $w \not\in \dom(\Phi(A_n))$.} This case is similar to the
previous one. However, during such a step, the frame of each configuration is
enriched, and thus the fact that $\Phi(A_{n+1}) \estat
\Phi(B_{n+1})$ is now a consequence of Proposition~\ref{prop:det-eint}.

\medskip{}

\noindent \emph{Case $\alpha_{n+1} = \Ses(a,\vect{c})$ for some $a$, and
  some $\vect{c}$.}
Firstly, we show for later that $\vect c$ are fresh in $B_{n}$.
Indeed, we deduce from $\bc(\alpha_1.\ldots\alpha_{n+1})\cap\fc(B)=\emptyset$
that $\vect c$ are fresh in $B$ and we know that free channels of $B_n$
are included in $\fc(B)\cup\bc(\alpha_1.\ldots\alpha_{n})$.
Thereby, if there was a channel $c_i\in\vect c\cap\fc(B_n)$ it would
be in $\bc(\alpha_1.\ldots\alpha_{n})$ but this is forbidden because of the freshness condition 
(in the current trace) over channels,
\ie new channels cannot be introduced twice (once in $\alpha_1.\ldots\alpha_{n}$ and once in $\alpha_{n+1}$).

As before, we obtain
\[
A_n = \proc{
  \{\loc{\bang{a}{\vect{c}, \vect{n_P}}{P}}{\ell_{n+1}}\} \uplus \p_0
}{\Phi_0}
\quad\mbox{and}\quad
B_n  = \proc{
  \{\loc{\bang{a}{\vect{c}, \vect{n_Q}}{Q}}{\ell_{n+1}}\} \uplus \q_0
}{\Psi_0}
\]
for some $P$, $Q$, $\p_0$, $\q_0$, $\Phi_0$, and $\Psi_0$.
Moreover, since $\skl(A_n) = \skl(B_n)$, we deduce that
$\skl(\p_0) = \skl(\q_0)$.

We have
$$A_{n+1} = \proc{
  \{ \loc{P}{\ell_{n+1} 1},
  \loc{\bang{a}{\vect{c}, \vect{n_P}}{P}}{\ell_{n+1} 2} \}
  \uplus \p_0}{\Phi_0}$$
Accordingly, let us pose
\[
  B_{n+1} = \proc{
    \{ \loc{Q}{\ell_{n+1} 1},
       \loc{\bang{a}{\vect{c}, \vect{n_Q}}{Q}}{\ell_{n+1} 2} \}
    \uplus \q_0}{\Psi_0}
\]
We have
that: 
\begin{itemize}
\item $B_n \lrstep{\loc{\Ses(a,\vect{c})}{\ell_{n+1}}} B_{n+1}$; and
\item $\Phi(A_{n+1}) = \Phi(A_n) \estat \Phi(B_n) = \Phi(B_{n+1})$.
\end{itemize}
It remains to show that $\skl(A_{n+1}) = \skl(B_{n+1})$. Since we
have that $\skl(\p_0) = \skl(\q_0)$, and since the labels of
corresponding subprocesses are the same on both sides, we only need to show 
that:
\begin{itemize}
\item $\sk(P) = \sk(Q)$ and 
\item $\sk(\bang{a}{\vect{c}, \vect{n_P}}{P}) =
       \sk(\bang{a}{\vect{c}, \vect{n_Q}}{Q})$
\end{itemize}

As in the previous case, thanks to Proposition~\ref{prop:det-eint}, we know that
$\enab(A_{n+1}) = \enab(B_{n+1})$, and we deduce that
$\enab(P) = \enab(Q)$ and
$\enab(\bang{a}{\vect{c}, \vect{n_P}}{P}) =
       \enab(\bang{a}{\vect{c}, \vect{n_Q}}{Q})$,
which allows us to conclude using Proposition~\ref{prop:skfromo}.
\end{proof}


Finally, we can define the {\em annotated trace equivalence} and show that,
for action-deterministic configurations, it coincides with trace equivalence.

\begin{definition}
Let~$A$ and~$B$ be two configurations.
We have that $A \sqsubseteq_a B$ if,
for every $A'$ such that $A \sinta{\tr} A'$ with
$\bc(\tr)\cap\fc(B)=\emptyset$, 
then there exists $B'$ such that $B\sinta{\tr} B'$, and $\Phi\statequiv \Psi$.
They are in \emph{annotated trace equivalence}, denoted $A \einta B$, if 
$A \sqsubseteq_a B$ and $B \sqsubseteq_a A$.
\label{def:einta}
\end{definition}

\begin{lemma}
  Let $A$ and $B$ be two action-deterministic configurations such that $\skl(A) = \skl(B)$.
  \[
  A\eint B\text{ if, and only if, }A\einta B
  \]
\label{lem:eint-einta}
\end{lemma}
\begin{proof}
Firstly, $A \einta B$ trivially implies $A\eint B$.
For the other direction, we use Lemma~\ref{lem:strong-symmetry} to conclude.
\end{proof}

\section{Compression}
\label{sec:app:compression}

\subsection{Reachability}

\restatelemma{thm:reach-comp}

\begin{proof}
Let $A = \proc{\mathcal{P}}{\Phi}$ be a configuration and
 $\proc{\mathcal{P}}{\Phi}\sinta{\tr}A'$ a complete execution. 
We proceed by induction on the length of $\tr$, distinguishing two cases.

\smallskip{}

\noindent\emph{Case 1.} We first consider the case where there is
at least one process in $\mathcal{P}$ that is negative and 
non-replicated. Since we are considering a complete execution,
at least one negative action $\alpha$ is performed on this process in $\tr$.
This action may be an output, the decomposition of a parallel composition,
or the removal of a zero.
If there are more than one such action, we choose the one that can be performed
using {\sc Neg}, \ie the one that is minimal according to our arbitrary order on labelled skeletons.
Since our action can be performed initially by our process, and by 
well-labelling, the label of the action is independent with all labels of
previously executed actions in $\tr$.
Moreover, there cannot be any second-order dependency between $\alpha$
and one of those actions. Indeed,  if $\alpha$ is an output, no input
performed before $\alpha$ is able to use the handle of $\alpha$.
It can thus be swapped before all the others by using Lemma~\ref{lem:perm},
obtaining an execution of 
trace $\alpha.\tr'$ ending in the same configuration $A'$.
The rule {\sc Neg} can be performed in the compressed 
semantics to trigger the action $\alpha$, and by induction hypothesis on $\tr'$ we can complete our 
compressed execution towards $A'$.

\smallskip{}

\noindent\emph{Case 2.} Otherwise, when $\mathcal{P}$ contains only positive or replicated processes,
we must choose one  process to focus on, start a positive 
phase and execute all its actions until we can finally release the focus.
As long as all  processes are positive or replicated, only input or 
session actions can be performed. In either case, the action yields a new 
process (the continuation of the input, or the new session)
which may be negative or positive.
We define the \emph{positive prefix} of our execution as the prefix of
actions for which all but the last transition yield a positive 
process.
It is guaranteed to exist because $A'$ contains only negative
processes.

The positive prefix is composed only of input and replication actions.
Because session actions are performed by negative processes, and no new
negative  process is created in the positive prefix, session actions can
be permuted at the beginning of the prefix thanks to Lemma~\ref{lem:perm}.
Thus, we assume without loss of
generality that the prefix is composed of session actions, followed by input
actions: we write $\tr = \tr_{!} . \tr_\mathrm{in} . \tr_0$.
In the portion of the execution where inputs of $\tr_\mathrm{in}$ are performed,
there is an obvious bijective mapping between the processes of any 
configuration and its successor, allowing us to follow execution threads,
and to freely permute inputs pertaining to different threads.
Such permutations are made possible by Lemma~\ref{lem:perm}.
Indeed,
they concern actions that are (i) sequentially independent (\ie labels are independent)
since two different threads involve actions performed by processes in parallel
and (ii) recipe independent since there is no output action in $\tr_{\mathrm{in}}$.

The last action of the positive prefix releases a negative process
$P^{-}$. Let $P$ be its antecedent (trough its corresponding thread)
in the configuration obtained after the
execution of $\tr_{!}$. We have that:
\[ A = \proc{\mathcal{P}}{\Phi}
  \sinta{\tr_{!}} \proc{\mathcal{P}_1\uplus\{P\}}{\Phi}
  \sinta{\tr_\mathrm{in}} \proc{\mathcal{P}_2\uplus\{P^{-}\}}{\Phi}
  \sinta{\tr_0} A' \]
Now we can write $\mathcal{P} = \mathcal{P}_0 \uplus \{ P_\mathsf{f} \}$
where $P_\mathsf{f}$ is either $P$ or a replicated  process that gives rise
to $P$ in one transition. By permuting actions pertaining to $P_\mathsf{f}$
before all others thanks to Lemma~\ref{lem:perm} and previous remarks,
 we obtain an execution of the form
\[ A = \proc{\mathcal{P}_0 \uplus \{ P_\mathsf{f} \}}{\Phi}
  \sinta{\tr_1} \proc{\mathcal{P}'_0\uplus\{P^{-}\}}{\Phi}
\sinta{\tr_2} \proc{\mathcal{P}_2\uplus\{P^{-}\}}{\Phi} \sinta{\tr_0} A' \]
where $\tr_1 . \tr_2 . \tr_0$ is a permutation of $\tr$ of
independent labelled actions,
$\tr_1 =\loc{\alpha}{\ell}.\tr'_1$,
and
$\p'_0 = \mathcal{P}_0$ when $P_\mathsf{f} = P$,
and $\p'_0 = \mathcal{P}_0 \uplus \{P_\mathsf{f}\}$ otherwise.

In the compressed semantics, if we initiate a focus on $P_\mathsf{f}$ we can execute
the actions of $\tr_1$ and release the focus when reaching $P^{-}$, \ie
we have that (where $\ell'$ is the label of $P^{-}$):
\[ \procc{\p}{\wfoc}{\Phi} \sintc{\loc{\Foc(\alpha)}{\ell}} \sintc{\tr'_1}
\procc{\p'_0}{P^{-}}{\Phi} \sintc{\loc{\Rel}{\ell'}} \procc{\p'_0\uplus\{P^{-}\}}{\wfoc}{\Phi}. \]
We can conclude by induction hypothesis on $\tr_2 . \tr_0$.
\end{proof}

\subsection{Equivalence}
We prove below the two implications of Theorem~\ref{thm:eintc},
dealing first with soundness and then
with the more involved completeness result.

\begin{lemma}[Soundness]
  Let $A$ and $B$ be action-deterministic 
  configurations such that $\skl(A) = \skl(B)$.
  We have that $A
  \eint B$ implies $\foc{A} \eintc \foc{B}$.
\label{lem:eintc-sound}
\end{lemma}

\begin{proof}
  By symmetry it suffices to show $\foc{A} \sqsubseteq_c \foc{B}$. Consider an execution 
  $\foc{A} \sintc{\tr}A'$ such that $\bc(\tr)\cap\fc(B)=\emptyset$.
  Thanks to Lemma~\ref{lem:reach-sound}, we know that $A \sinta{\defoc{\tr}}
  \defoc{A'}$. Let $\defoc{\tr} = \loc{\alpha_1}{\ell_1} \ldots \loc{\alpha_n}{\ell_n}$,
  and we denote $A_1, \ldots, A_n$ the intermediate configurations
  that are reached during this execution. We have that:
  \[
  A_0 = A \sinta{\loc{\alpha_1}{\ell_1}} A_1 \sinta{\loc{\alpha_2}{\ell_2}}
  \ldots \sinta{\loc{\alpha_n}{\ell_n}} A_n = \proc{\p}{\Phi_A}.
  \]

  Applying Lemma~\ref{lem:strong-symmetry}, we
  deduce that $B$ can perform a very similar execution (same labels,
  same actions), \ie
  \[
  B_0 = B \sinta{\loc{\alpha_1}{\ell_1}} B_1 \sinta{\loc{\alpha_2}{\ell_2}}
  \ldots \sinta{\loc{\alpha_n}{\ell_n}} B_n = \proc{\q}{\Phi_B}.
  \]
  with $\Phi(A_i) \estat \Phi(B_i)$ and $\skl(A_i) = \skl(B_i)$ for $0 \leq i\leq n$.

  Due to this strong symmetry, we are sure that $\foc{B}$ will be able
  to do this execution in the compressed semantics. In particular, the
  fact that a given configuration $B_i$ can start a positive phase or has to release
  the focus is determined by the set $\skl(B_i)=\skl(A_i)$ and the fact that it can keep the focus
  on a specific process while performing positive actions can be deduce from labels of $\tr$.
  Finally, we have shown that if $A_i$ can execute an action $\alpha$ using {\sc Neg} rule
  then $B_i$ can as well.
  The only missing part is about the fact that {\sc Neg} has been made non-branching using
  an arbitrary order on labelled skeletons.
  Let say we can use {\sc Neg} only for actions
  whose skeleton is minimal among others skeletons of available, negative actions.
  Using $\skl(A_i)=\skl(B_i)$, we easily show that this is symmetric for $A_i$
  and $B_i$.
  This way, we obtain an execution $\foc{B}\sintc{\tr}B'$ with
  $\defoc{B'}=B_n$.
  Finally, we have $\Phi(B')=\Phi_B\estat\Phi_A=\Phi(A')$.
\end{proof}


\begin{lemma}
  Let $A$ and $B$ be two action-deterministic configurations such that
  $A\eintc B$.  If $A\sintc{\tr}A'$ and $B\sintc{\tr}B'$ for a
  labelled trace $\tr$ then $A'\eintc B'$.
  \label{lem:onestepc}
\end{lemma}

\begin{proof}
  We assume  $A\eintc B$, $A\sintc{\tr}A'$ and $B\sintc{\tr}B'$ for a labelled trace
  $\tr=\alpha_1.\ldots\alpha_n$. We shall prove $A'\eintc B'$.
  By symmetry, we show one inclusion. Consider an execution
  $A'\sintc{\tr_2}A_2$ such that $\bc(\tr_2)\cap\fc(B')=\emptyset$. Let us construct
  an execution $B'\sintc{\tr_2}B_2$ such that $\Phi(A_2)\estat\Phi(B_2)$.
  Firstly, remark that since $B\sintc\tr B'$,
  we have that $\bc(\tr)\cap\fc(B)=\emptyset$.
  In order to exploit our hypothesis $A\eintc B$, we shall prove that
  $\bc(\tr.\tr_2)\cap\fc(B)=\emptyset$.
  All that remains to show is   $\bc(\tr_2)\cap(\fc(B)\backslash\fc(B'))=\emptyset$.
  This is implied by the fact that channels in $\fc(B)\backslash\fc(B')$
  must occur in $\tr$ as bound channels and, because of the execution
  $A\sintc{\tr.\tr_2}A_2$, those channels cannot appear in the set $\bc(\tr_2)$.

  We have that $A\sintc{\tr.\tr_2}A_2$ and
  thus by hypothesis, $B\sintc{\tr.\tr_2}B_2$ with
  $\Phi(B_2)\estat\Phi(A_2)$.
  Since there can be at most one process in $B$
  that has a label that matches the label of $\alpha_1$,
  there is at most one configuration $B_1$ that satisfies
  $B\sintc{\alpha_1}B_1$. By iteration, we obtain the unicity
  of the configuration $B'$ satisfying $B\sintc{\tr}B'$.
  We thus have obtained the required execution
  $B'\sintc{\tr_2}B_2$.
\end{proof}

\begin{lemma} 
  Let $A$ and $B$ be two action-deterministic configurations.  If
  for any complete execution $A\sinta{\tr}A'$ with $\bc(\tr)\cap\fc(B)=\emptyset$,
  there exists a trace $\tr'$ and an execution $B\sinta{\tr'}B'$ such
  that $\Phi(A')\estat\Phi(B')$, then
  $A\sqsubseteq_a B$.
  \label{lem:completion}
\end{lemma}
\begin{proof}
  Let $A\sinta{\tr_0}A_0$ be an execution of $A$ with $\bc(\tr_0)\cap\fc(B)=\emptyset$.
  Firstly, we complete the latter execution in an arbitrary way $A\sinta{\tr_0.\tr_1}A'$
  such that any process of $A'$ is replicated and $\bc(\tr_0.\tr_1)\cap\fc(B)=\emptyset$
  (it suffices to choose fresh channels for $B$ as well).
  By hypothesis, there exists an execution 
  $B\sinta{\tr_0.\tr_1}B'$ such that  $\Phi(A')\estat\Phi(B')$.
  The latter execution of $B$ is thus of the form
  $B\sinta{\tr_0}B_0\sinta{\tr_1}B'$.
  It remains to show that $\Phi(A_0)\estat\Phi(B_0)$.
  For the sake of contradiction, we assume that $\Phi(B_0)\estat\Phi(A_0)$
  does not hold.  In other words, there is a test of equality over
  $\dom(\Phi(B_0))$ that holds for $\Phi(A_0)$ and does not for
  $\Phi(B_0)$ (or the converse). 
  Since $\dom(\Phi(B_0))\subseteq\dom(\Phi(B'))=\dom(\Phi(A'))$,
  this same test can be used to
  conclude that $\Phi(A')\estat\Phi(B')$ does not hold as well
  leading to an absurd conclusion.
\ignore{      
 We now distinguish two
  cases for the given witness: (i) there is no execution for $B$ with the same
  observable actions as $\tr_0$ at all (ii) or
  for all such execution, $\Phi(A_0)\estat\Phi(B_0)$ does not hold.
  In case (i): there is no execution for $B$ with the same
  observable action as $\tr_0.\tr_1$ as well.  In case
  (ii): consider a configuration $B'$ such that
  $B\sinta{\tr_0'.\tr_1'}B'$ with $\obs(\tr_0.\tr_1)=\obs(\tr_0'.\tr_1')$.
  If $B_0$ is such that
  \[B\sinta{\tr_0'}B_0\sinta{\tr_1'}B'\] then $\Phi(B_0)\estat\Phi(A_0)$
  does not hold.  In other words, there is a test of equality over
  $\dom(\Phi(B_0))$ that holds for $\Phi(A_0)$ and does not for
  $\Phi(B_0)$ (or the converse).  This same test can be used to
  conclude that $\Phi(A')\estat\Phi(B')$ does not hold as well.
}
\end{proof}

\begin{lemma}[Completeness]
\label{proposition:equiv-complete-compr}
Let $A$ and $B$ be two action-deterministic configurations satisfying
$\skl(A)=\skl(B)$.
Then $\foc A \eintc \foc B$ implies $A \eint B$.
\label{lem:eintc-complete}
\end{lemma}



\begin{proof}
  Assume $\defoc A\eintc \defoc B$,
  thanks to Lemma~\ref{lem:eint-einta}, it suffices to show
  $A\einta B$.
  Let us show the following intermediate result:
  {\em for any complete execution $A \sinta{\tr} A'$ such that
    $\bc(\tr)\cap\fc(B)=\emptyset$, there is an
    execution $B \sinta{\tr} B'$ such that  $\Phi(A') \statequiv \Phi(B')$.}
  Thanks to Lemma~\ref{lem:completion} and by symmetry of $\einta$,
  this intermediate result implies the required conclusion $A \einta B$.

  Let $A\sinta{\tr}A'$ be a complete execution with $\bc(\tr)\cap\fc(B)=\emptyset$.
  We thus have that $A'$ is initial.
  Applying the Lemma~\ref{thm:reach-comp}, we obtain
  a trace $\tr_c$ such that $\foc A  \sintc{\tr_c} \foc{A'}$
  and $\defoc{\tr_c}$ can be obtained from $\tr$ by swapping independent actions.
  Since we have $\foc A\eintc \foc B$, we deduce that
  $\foc B \sintc{\tr_c} \foc{B'}$ with $\Phi(A') \statequiv \Phi(B')$.
  Lemma~\ref{lem:reach-sound} gives us $B\sinta{\defoc{\tr_c}}B'$.
  We can now apply Lemma~\ref{lem:perm} to obtain
  $B\sinta{\tr} B'$ and conclude.
\end{proof}


\section{Reduction}
\label{sec:app:reduction}

\subsection{Reachability}

\restatelemma{lem:block-swap}

\begin{proof}
  Thanks to Lemma~\ref{lem:reach-sound}, we have that
  $\defoc{A}\sinta{\defoc\tr}\proc{\p}{\Phi}$.
  We first prove that $\tr'$ can be performed using $\sinta{}$.
  For this, it suffices to establish the implication for each of the two
  generators of $\eqtb{\Phi}$.
  The first case is given by Lemma~\ref{lem:perm}.
  The second one is a common property of (derivatives of)
  the applied $\pi$-calculus that follows from a simple observation
  of the transition rules.
  Finally, we must prove that $\tr'$ can be played using $\sintc{}$.
  Thanks  to initiality of $A$ and $\procc{\p}{\wfoc}{\Phi}$ we know that each
  block of $\tr$ starts when the configuration is initial and after performing it we get another initial configuration. 
  This is still true in $\defoc{A}\sinta{\defoc{\tr'}}\proc{\p}{\Phi}$.
  Finally, labels of blocks of $\tr'$ ensures that a single process is used in a positive
  part of any block.
  Having proving those two facts, we can easily show that each block of $\tr'$
  can be performed using $\sintc{}$.
\end{proof}

For the sake of readability of the following proofs, we now introduce some notations
(where $b_1,b_2$ are two blocks and $\Phi$ is a frame):
\begin{itemize}
\item we note $b_1\inpar^s b_2$ when $b_1$ and $b_2$ are
  sequentially independent (\ie for any $\alpha_1\in b_1$ and
  $\alpha_2\in b_2$, $\alpha_1$ and $\alpha_2$ are sequentially
  independent);
\item we note $b_1\inpar^d b_2$ when $b_1$ and $b_2$ are
  recipe independent (\ie for any $\alpha_1\in b_1$ and $\alpha_2\in
  b_2$, $\alpha_1$ and $\alpha_2$ are recipe independent);
\item  for two traces $\tr=b_1 \ldots b_n$ and
  $\tr'=b_1'\ldots b_m'$, we note $(\tr=_\E \tr')\Phi$ when $n=m$ and
  for all $i\in[1;n],$ $(b_i=_\E b_i')\Phi$.
\end{itemize}

\begin{remark}
  Let $A\sintc{\tr}A'$ be any compressed execution.
  If $b_1$ and $b_2$ are two blocks of $\tr$ and $a_1$ (resp. $a_2)$
  is the first labelled action of $b_1$ (resp. $b_2$) we have the following:
\[
b_1\inpar^s b_2 \iff a_1\text{ is sequentially independent with }a_2.
\]
This is implied by the fact that for any other action $a_1'$ of $b_1$, its label has
the label of $a_1$ as a prefix.
\end{remark}

\restatelemma{lem:min-swap}

\begin{proof}
Let $A$ and $\procc{\p}{\wfoc}{\Phi}$ be two configurations such that
$A \sintc{\tr} \procc{\p}{\wfoc}{\Phi}$. 

\medskip{}
$(\Rightarrow$)
We first show that if $\tr$ is $\Phi$-minimal, then $A \sintd{\tr}
\procc{\p}{\wfoc}{\Phi}$ by induction on the trace~$\tr$. The base
case, \ie $\tr = \epsilon$ is straightforward.
Now, assume that $\tr=\tr_0.b$ for some block $b$
and $A\sintc{\tr_0} \procc{\p_0}{\wfoc}{\Phi_0}\sintc{b}A'$.
Since $\tr$ is $\Phi$-minimal, we also have that $\tr_0$ is
$\Phi_0$-minimal and thus we obtain by induction hypothesis that
$A\sintd{\tr_0} \procc{\p_0}{\wfoc}{\Phi_0}$. To conclude, 
it remains to show that  $\constrd{b'}{\tr_0}$ for any $b'$ such that
$(b'=_\E b)\Phi_0$. Assume that it is not the case, this means that
for some $b'$ such that $(b=_\E b')\Phi_0$, the trace $\tr_0$ can
be written
$\tr'_0.b_0. \ldots.b_n$ with:
\begin{center}
$b_i \inpar b'$ and $b_i \prec b'$ for any $i>0$, as well as 
$b_0 \inpar b'$ and $b' \prec b_0$.
\end{center}
Let $\tr' = \tr'_0.b'.b_0 \ldots b_n$. We have
$\tr' \ordlex \tr$ and $\tr' \eqtb{\Phi} \tr_0.b'$, which
contradicts the $\Phi$-minimality of $\tr$.



$(\Leftarrow)$
Now, we assume that $\tr$ is not $\Phi$-minimal, and 
we want to establish that $\tr$ cannot be executed in the reduced
semantics.
Let $\tr_m$ be the $\Phi$-minimal trace of the equivalence class of $\tr$.
We have in particular $\tr_m\eqtb\Phi \tr$ and $\tr_m\ordlex\tr$.
Now, we let $\tr_m^s$ (resp. $\tr^s$) be the
``trace of labelled skeletons'' associated to $\tr_m$ (resp. $\tr$).
Let $\tr_0^s$ be the longest common prefix of $\tr_m^s$ and $\tr^s$
and $\tr_0$ (resp. $\tr_0'$) be the corresponding prefix of $\tr$ (resp. $\tr_m$).
We have a decomposition of the form
$\tr=\tr_0.b.\tr_1$ and $\tr_m=\tr_0'.b_m.\tr_1'$ with 
$(\tr_0=_\E\tr_0')\Phi$ and $b_m\prec b$.
Again, since when dropping recipes, the relation $\eqtb{\Phi}$ only swaps
sequentially independent labelled skeletons, block $b_m$ must have
a counterpart in $\tr$ and, more precisely, in $\tr_1$. We thus have 
a more precise decomposition of $\tr$:
$\tr=\tr_0.b.\tr_{11}.b_m'.\tr_{12}$ such that $(b_m'=_\E b_m)\Phi$.

We now show  that $b_m'$ cannot
be executed after $\tr_0.b.\tr_{11}$ in the reduced semantics
(assuming that the trace has been executed so far in
the reduced semantics). 
In other words,
we show that $\constrd{b_m'}{\tr_0.b.\tr_{11}}$ does \emph{not}
hold.
We have seen that $(b_m' =_\E b_m)\Phi$ and $b_m\prec b$ so it suffices to
show:
\begin{center}
$b_m \inpar b_i$ for any $b_i \in b.\tr_{11}$
\end{center}
First, we prove $b_i \inpar^d b_m$ (\ie they are recipe independent)
for any $b_i \in b.\tr_{11}$.
This comes from the fact that $\tr'_0.b_m.\tr'_1=\tr_m$ is plausible,
and thus the inputs of $b_m$ only use handles introduced in $\tr'_0$ which are the same as those
introduced in $\tr_0$. In particular, the inputs of $b_m$ do not rely
on the handles introduced in $b.\tr_{11}$.
Similarly, using the fact that $\tr_0.b.\tr_{11}.b_m'.\tr_{12}=\tr$ is plausible and
$b_m'^-=b_m^-$,
we deduce that handles of outputs of $b_m$ are not used in $b.\tr_{11}$.

Second, we show that $b_i \inpar^s b_m$ (\ie they are sequentially independent) 
for any $b_i \in b.\tr_{11}$.
For this, we remark that for any traces $\tr_1.b.\tr_2\eqtb{\Phi}\tr_1'.b'.\tr_2'$
such that $(b=_\E b')\Phi$, we have that $b\inpar^s b_s$ for all $b_s\in\skl(\tr_1')\backslash\skl(\tr_1)$ where
$\skl(\tr)$ is the multiset of labelled skeletons of blocks of $\tr$,
and $\backslash$ should be read as multiset subtraction.
This can be easily shown by induction on the relation $\eqtb{\Phi}$.
By applying this
helping remark to $\tr_0.b_m'.\tr_2'\eqtb{\Phi}\tr_0.b.\tr_{11}.b_m'.\tr_{12}$,
we obtain the required conclusion: $b_m'\inpar^s b.\tr_{11}$ and thus $b_m\inpar^s b.\tr_{11}$.
\end{proof}

\subsection{Equivalence}

\restatableproposition{prop:eqtb-sym}{
  For any static equivalent frames $\Phi\estat\Psi$ and compressed
  traces $\tr$ and $\tr'$,
  we have that $\tr \eqtb{\Phi} \tr'$ if, and only if,
  $\tr \eqtb{\Psi} \tr'$.
}
\begin{proof}
The two implications are symmetric, we thus only show one implication.
Considering the two generators of $\eqtb{\Phi}$, the only non-trivial
step is to show that $\tr . b_1 . \tr' \eqtb{\Psi} \tr . b_2 . \tr'$
when $(b_1 =_\E b_2)\Phi$. But the latter condition, together with
$\Phi\estat\Psi$, yields $(b_1 =_\E b_2)\Psi$ which allows us to
conclude.
\end{proof}

\begin{lemma}
\label{lem:completion-comp}
Let $A$ and $B$ be two action-deterministic
configurations. If for any complete execution of the form $A \sintc{\tr} \procc{\p}{\wfoc}{\Phi}$ 
with $\bc(\tr)\cap\fc(B)=\emptyset$,
there exists an execution $B \sintc{\tr} \procc{\q}{\wfoc}{\Psi}$ such that $\Phi
\estat \Psi$, then $A \sqsubseteq_c B$.
\end{lemma}

\begin{proof}
Let $A$ and $B$ be two action-deterministic
configurations, and assume that for any complete execution 
$A \sintc{\tr} \procc{\p}{\wfoc}{\Phi}$ with $\bc(\tr)\cap\fc(B)=\emptyset$,
there exists an execution $B \sintc{\tr} \procc{\q}{\wfoc}{\Psi}$
such that $\Phi
\estat \Psi$. Now, we have to establish that $A \sqsubseteq_c B$.

Let $\procc{\p'}{\wfoc}{\Phi'}$ be a configuration such that $A \sintc{\tr'} \procc{\p'}{\wfoc}{\Phi'}$. First,
we can complete this execution to reach a process $
\procc{\p}{\wfoc}{\Phi}$ such that each process $P \in \p$ is
replicated, \ie 
\begin{center}
$A \sintc{\tr'} \procc{\p'}{\wfoc}{\Phi'} \sintc{\tr^+}
\procc{\p}{\wfoc}{\Phi}$ is a complete execution.
\end{center}
Without loss of generality, we can choose $\tr^+$ so that it satisfies
$\bc(\tr^+)\cap\fc(B)=\emptyset$.
By hypothesis, we know that there exists an execution $B \sintc{\tr'
  \tr^+} \procc{\q}{\wfoc}{\Psi}$ such that $\Phi \estat \Psi$.
Let $B'$ be the configuration reached along this execution after the
execution of $\tr'$ and $\Psi'$ its frame. 
%
Similarly to the proof of Lemma~\ref{lem:completion}, we prove that
$\Phi\estat\Psi$ implies $\Phi'\estat\Psi'$.
\end{proof}

\restatetheorem{thm:redequiv}

\begin{proof}
We prove the two directions separately.
\medskip{}

\noindent
($\Rightarrow$) \emph{$A \sqsubseteq_c B$ implies $A \sqsubseteq_r B$.}
Consider an execution of the form $A \sintd{\tr}
\procc{\p}{\wfoc}{\Phi}$ with $\bc(\tr)\cap\fc(B)=\emptyset$.
Using Lemma~\ref{lem:reach-red-sound}, we know that $A \sintc{\tr}
\procc{\p}{\wfoc}{\Phi}$, and Lemma~\ref{lem:min-swap} tells us that
$\tr$ is $\Phi$-minimal. Since $A \sqsubseteq_c B$, we deduce that
there exists $\procc{\q}{\wfoc}{\Psi}$ such that:
\begin{center}
$B \sintc{\tr} \procc{\q}{\wfoc}{\Psi}$ and $\Phi \estat \Psi$.
\end{center}
Now, by Proposition~\ref{prop:eqtb-sym} we obtain that $\tr$
is also $\Psi$-minimal, and so Lemma~\ref{lem:min-swap} tells us
that the execution of $\tr$ by $B$ can also be performed in
the reduced semantics.

\medskip{}

\noindent
($\Leftarrow$) \emph{$A \sqsubseteq_r B$ implies $A \sqsubseteq_c B$.}
Relying on Lemma~\ref{lem:completion-comp}, it is actually sufficient
to show that for any complete execution 
$A \sintc{\tr} \procc{\p}{\wfoc}{\Phi}=A'$
with $\bc(\tr)\cap\fc(B)=\emptyset$, there
exists an execution of the form $B \sintc{\tr}
\procc{\q}{\wfoc}{\Psi}$ such that $\Phi \estat \Psi$.
Note that since the given execution is complete, we have that $A'$ is initial.
Let $\tr'$ be a $\Phi$-minimal trace in the equivalence class of $\tr$.
We have that $A$ executes $\tr'$ in the reduced semantics, and so for
some $B'$ we have
\[ B \sintd{\tr'} B' \mbox{~ and ~} \Phi(B') \estat \Phi. \]
Using Lemma~\ref{lem:reach-red-sound}, we obtain the same execution
in the compressed semantics.
Finally, by Proposition~\ref{prop:eqtb-sym} we obtain
$\tr' \eqtb{\Phi(B')} \tr$, and by Lemma~\ref{lem:block-swap} we
obtain:
\[ B \sintc{\tr} B' \mbox{~ and ~} \Phi(B') \estat \Phi \]
\end{proof}

\section{Optimization for improper blocks}
\label{sec:impropre}

Between any two initial configurations, the compressed as well as the reduced semantics execute a sequence of actions of the form 
 of the form
\;$\Foc(\alpha).\tr^+.\Rel.\tr^-$\;
where~$\tr^+$ is a (possibly empty) 
sequence of input actions, whereas $\tr^-$ is a (non-empty) sequence
of $\Out$, $\Para$, and $\Zero$ actions. 
Such sequences are called \emph{blocks}. Now, we also make a
distinction between blocks having a null negative phase (\ie
$\tr^- = \Zero$), and the
others. The former are called \emph{improper blocks} whereas the latter are
called \emph{proper blocks}.
Finally, we say that a trace is \emph{proper} if it contains at most one
improper block and only at the end of the trace.

The idea is that, when checking trace equivalence, we do not have to consider
all possible traces but we can actually restrict to proper ones.
We present below an improved version of the notions of compressed and
reduced equivalence.


\subsection{Compression}
\label{subsec:compression}

\begin{definition}
Let~$A$ and~$B$ be two configurations. We have that $A \sqsubseteq_{c+i} B$ if,
for every $A'$ such that $A \sintc{\tr} A'$ and
$\bc(\tr)\cap\fc(B)=\emptyset$ for some \emph{proper} $\tr$,
there exists $B'$ such that $B\sintc{\tr} B'$ with $\Phi(A')\statequiv \Phi(B')$.
We write $A \approx_{c+i} B$, if 
$A \sqsubseteq_{c+i} B$ and $B \sqsubseteq_{c+i} A$.
\label{def:eintc-impropre}
\end{definition}

Operationally, we can obtain $\approx_{c+i}$ by adding
a case to the {\sc Release} rule (and constraining the former rule
to not apply in that case):
\[
\begin{array}{lc}
  \mbox{\sc Release$_i$} &  \infer[]
  {\trip{\mathcal{P}}{\loc{0 }{\ell}}{\Phi}\fsintc{\;\loc{\Rel}{\ell}\;} 
    \trip{\emptyset}{\wfoc}{\Phi}}
  \\
\end{array}
\]
This rule discards exactly the traces containing an improper block
that is not at the end: note that having $0$ under focus implies
that the negative part of the block will be restricted to $\Zero$.
This is because no negative 
actions was available before performing this block and consequently there can
only be a $\Zero$ in the negative part of the block.

\begin{proposition}
Let $A$ and $B$ be two initial action-deterministic configurations.
\begin{center}
$A \eintc B$ if, and only, if, $A \approx_{c+i} B$
\end{center}
\end{proposition}

\begin{proof}
The $(\Rightarrow)$  direction is trivial. We focus on the other one.
Assume that $A \sqsubseteq_{c+i} B$. Let $A'$ be such
that $A \sintc{\tr} A'$ for some $\tr$ such that
$\bc(\tr)\cap\fc(B)=\emptyset$.
Let $b_1, \ldots, b_k$ be the improper blocks that occur in $\tr$. We
have that $b_1,\ldots, b_k$ are pairwise independent. Moreover, 
we have that there exists~$\tr'$ made of proper blocks such that
$\tr'.b_1.\ldots.b_k$ is obtained from $\tr$ by swapping independent
blocks, and thus
we have that $A \sintc{\tr'.b_1.\ldots.b_k} A'$.
There exist initial configurations
$A_0, A_1,\ldots,A_k$ such that
$A \sintc{\tr'} A_0$, and 
\[
A_0 \sintc{b_1} A_1, \; A_0 \sintc{b_2} A_2, \; \ldots, A_0 \sintc{b_k} A_k.
\]
 Thanks to our hypothesis, we deduce that there exist configurations $B_0$, $B_1,
 \ldots, B_k$ such that
$B \sintc{\tr'} B_0$ with $\Phi(A_0) \statequiv \Phi(B_0)$, and 
\[
B_0 \sintc{b_1} B_1, \; B_0 \sintc{b_2} B_2, \; \ldots, B_0 \sintc{b_k} B_k.
\]
We have also that $\Phi(A_0) \statequiv \Phi(B_0)$.
 Since blocks $b_1,\ldots,b_k$ are pairwise independent, we deduce that
 there exist $B'$ such that $B \sintc{\tr'} B_0 \sintc{b_1,\ldots,b_k}
 B'$ with $\Phi(B') = \Phi(B_0)$.
Then, permutations of blocks can be undone to retrieve $\tr$ (since
swapping have been done between independent blocks).
\end{proof}


\subsection{Reduction}
\label{subsec:reduction}

\begin{definition}
Let~$A$ and~$B$ be two configurations. We have that $A \sqsubseteq_{r+i} B$ if,
for every $A'$ such that $A \sintd{\tr} A'$ and
$\bc(\tr)\cap\fc(B)=\emptyset$ for some \emph{proper} $\tr$,
there exists $B'$ such that $B\sintd{\tr} B'$ with $\Phi(A')\statequiv
\Phi(B')$.
We write $A \approx_{r+i} B$, if 
$A \sqsubseteq_{r+i} B$ and $B \sqsubseteq_{r+i} A$.
\label{def:eintd-impropre}
\end{definition}

\begin{proposition}
Let $A$ and $B$ be two initial action-deterministic configurations.
\begin{center}
$A \eintd B$ if, and only, if, $A \approx_{r+i} B$
\end{center}
\end{proposition}

\begin{proof}
The  ($\Rightarrow$) direction is trivial.
We focus on the other one.
Assume that $A \sqsubseteq_{r+i} B$. Let $A'$ be such that $A
\sintd{\tr} A'$ for some $\tr$ such that
$\bc(\tr)\cap\fc(B)=\emptyset$.
Lemma~\ref{lem:reach-red-sound} tells us that $A
\sintc{\tr} A'$, and thanks to Lemma~\ref{lem:min-swap}, we have that
$\tr$ is $\Phi(A')$-minimal.

Let $b_1, \ldots, b_k$ be the improper blocks that occur in $\tr$. 
We have that there exist $\tr_0, \tr_1, \ldots,\tr_k$ made of proper blocks such that
$\tr = \tr_0.b_1.\tr_1.b_2.\ldots\tr_{k-1}.b_k.\tr_k$.  We
have that $b_1,\ldots, b_k$ are pairwise independent, and also:
\[ \tr \eqtb{\Phi(A')} \tr_0.\tr_1\ldots \tr_k. b_1.\ldots.b_k \]

Because there are no dependencies between
$b_i$ and $\tr_j$ for $i<j$, and because the $b_i$ do not have any output,
we have that 
$A \sintd{\tr_0.\tr_1.\ldots.\tr_k} A_0$, and also that:
\begin{center}
$A \sintd{\tr_0.b_1} A_1$, 
$A \sintd{\tr_0.\tr_1.b_2} A_2$,
$\ldots$,
$A \sintd{\tr_0.\tr_1. \ldots. \tr_{k-1}.b_k} A_k$.
\end{center}
Thanks to our hypothesis, we deduce that there exist $B_0, B_1,
 \ldots, B_k$ such that
$B \sintd{\tr_0.\tr_1.\ldots.\tr_k} B_0$ with $\Phi(A_0) \statequiv \Phi(B_0)$, and also that:
\begin{center}
$B \sintd{\tr_0.b_1} B_1$, 
$B \sintd{\tr_0.\tr_1.b_2} B_2$,
$\ldots$,
$B \sintd{\tr_0.\tr_1. \ldots. \tr_{k-1}.b_k} B_k$.
\end{center}

We deduce that there exists $B'$ such that 
$B
\sintc{\tr_0.\tr_1\ldots.\tr_k.b_1. \ldots. b_k} B'$.
Next, we observe that $\Phi(A') = \Phi(A_0) \statequiv \Phi(B_0) = \Phi(B')$.
From this we conclude
$\tr \eqtb{\Phi(B')} \tr_0.\tr_1\ldots \tr_k. b_1.\ldots.b_k$,
hence $B \sintc{\tr} B'$.
Since $\tr$ is $\Phi(A')$-minimal 
it is also $\Phi(B')$-minimal, and thus $B \sintd{\tr} B'$
by Lemma~\ref{lem:min-swap}.
\end{proof}

\end{document}